\pdfoutput=1
\documentclass[journal]{IEEEtran}
\usepackage[utf8]{inputenc}
\usepackage{amsmath}
\usepackage{amsfonts}
\usepackage{bbding}
\usepackage{amssymb}
\usepackage{array}
\usepackage{multirow}
\usepackage{graphicx}
\usepackage{algorithm}
\usepackage{algorithmic}
\newcommand\algorithmicprocedure{\textbf{procedure}}
\newcommand{\algorithmicendprocedure}{\algorithmicend\ \algorithmicprocedure}
\makeatletter
\newcommand\PROCEDURE[3][default]{%
  \ALC@it
  \algorithmicprocedure\ \textsc{#2}(#3)%
  \ALC@com{#1}%
  \begin{ALC@prc}%
}
\newcommand\ENDPROCEDURE{%
  \end{ALC@prc}%
  \ifthenelse{\boolean{ALC@noend}}{}{%
    \ALC@it\algorithmicendprocedure
  }%
}
\newenvironment{ALC@prc}{\begin{ALC@g}}{\end{ALC@g}}
\makeatother

\usepackage{color}
\usepackage{cite}

\allowdisplaybreaks

\DeclareMathOperator{\sgn}{\bf sgn}

\DeclareMathOperator{\svd}{\bf svd}
\DeclareMathOperator{\expt}{\mathbb{E}}

\DeclareMathOperator{\diag}{\bf diag}

\usepackage{amsthm}
\usepackage{stfloats}
\usepackage{caption}
\usepackage{subcaption}

\newtheorem{lemma}{Lemma}

\begin{document}
\title{Graph Representation Learning for Contention and Interference Management in Wireless Networks}

\author{
  \IEEEauthorblockN{Zhouyou Gu, Branka Vucetic, Kishore Chikkam, Pasquale Aliberti, Wibowo Hardjawana}\\
    \thanks{This work is supported in part by the University of Sydney's External Research Collaboration Seed Funding (Morse Micro), DVC Research, 2022, in part by an Australian Government Research Training Program Scholarship, in part by the University of Sydney's Supplementary and Completion Scholarships. The work of Zhouyou Gu was also supported in part by the Australian Research Council (ARC) Discovery grant number DP210103410 and the work of Branka Vucetic by the ARC Laureate Fellowship grant number FL160100032 and Discovery grant number DP210103410. The work of Wibowo Hardjawana was supported in part by the ARC Discovery grant number DP210103410.
    \emph{(Corresponding author: Zhouyou Gu.)}}
    \thanks{Z. Gu, B. Vucetic and  W. Hardjawana are with the School of Electrical and Information Engineering, the University of Sydney, Sydney, NSW 2006, Australia (email: \{zhouyou.gu, branka.vucetic, wibowo.hardjawana\}@sydney.edu.au).}
    \thanks{K. Chikkam and P. Aliberti are with Morse Micro, Sydney, Australia (email: \{kishore, pasquale.aliberti\}@morsemicro.com).}
    \thanks{Source codes are available at {https://github.com/zhouyou-gu/ac-grl-wi-fi}}
    }
\maketitle
\begin{abstract}
Restricted access window (RAW) in Wi-Fi 802.11ah networks manages contention and interference by grouping users and allocating periodic time slots for each group's transmissions. We will find the optimal user grouping decisions in RAW to maximize the network's worst-case user throughput. We review existing user grouping approaches and highlight their performance limitations in the above problem. We propose formulating user grouping as a graph construction problem where vertices represent users and edge weights indicate the contention and interference. This formulation leverages the graph's max cut to group users and optimizes edge weights to construct the optimal graph whose max cut yields the optimal grouping decisions. To achieve this optimal graph construction, we design an actor-critic graph representation learning (AC-GRL) algorithm. Specifically, the actor neural network (NN) is trained to estimate the optimal graph's edge weights using path losses between users and access points. A graph cut procedure uses semidefinite programming to solve the max cut efficiently and return the grouping decisions for the given weights. The critic NN approximates user throughput achieved by the above-returned decisions and is used to improve the actor. Additionally, we present an architecture that uses the online-measured throughput and path losses to fine-tune the decisions in response to changes in user populations and their locations. Simulations show that our methods achieve $30\%\sim80\%$ higher worst-case user throughput than the existing approaches and that the proposed architecture can further improve the worst-case user throughput by $5\%\sim30\%$ while ensuring timely updates of grouping decisions.
\end{abstract}
\section{Introduction}
With increasing demands for wireless connectivity, the IEEE task group ah has developed a dedicated standard, namely IEEE 802.11ah  (or Wi-Fi HaLow), providing low-power, long-range wireless connections to network users \cite{adame2014ieee,80211ah}.
Like most 802.11 families, Wi-Fi HaLow users rely on carrier-sense multiple access with collision avoidance (CSMA/CA) for channel access.
Specifically, when a user senses other users are transmitting on the channel, it will wait until the channel is free and further backoff for a random time before transmitting its packet to avoid packet collisions with other users.

Two issues exist in the Wi-Fi HaLow's channel access scheme \cite{jiang2007improving}.
First, when the user senses many other users' transmissions, e.g., geographically surrounded by many users, it will experience a significant waiting and backoff time before each transmission, triggered by the surrounding users' transmissions. Hence, the user can hardly make its transmissions.
Second, when two users cannot sense each other's transmissions due to large distance separation, i.e., hidden from each other, they can make concurrent transmissions, causing inter-user interference at the receiving access points (APs). Thus, the APs can barely decode transmissions due to low signal-to-interference-plus-noise ratios (SINR). 
We refer to neighboring users who trigger backoff or cause interference as the user's contending or hidden users, respectively.
The above two issues can cause the worst-case users to suffer from throughput starvation.
This creates outages in applications requiring service continuity, including remote healthcare and environmental hazard monitoring (e.g., for flood or earthquake) \cite{aust2015outdoor}, as the worst-case users cannot return their data to APs due to throughput starvation.

The Wi-Fi HaLow network introduces the restricted access window (RAW) mechanism \cite{80211ah,tian2017real} to reduce the contention and interference among users.
Specifically, users are divided into multiple groups in RAW. APs periodically schedule a RAW time slot for each user group's transmissions. 
Subsequently, users within each group only contend for packet transmissions during their assigned slots.
This means that transmissions from users of other groups do not trigger the waiting time or backoffs, thereby optimizing the time available for transmissions.
Also, the interference between two hidden users is eliminated if they are assigned to different RAW time slots.
The open literature has not studied how to design a user grouping scheme in scheduling RAW slots to avoid throughput starvation. Several existing approaches can potentially be applied to design user grouping schemes in RAW as follows.

\subsection{Related Works}\label{subsec:ac_grl_related_works}
\subsubsection{Markov-model-based Approach}
Research has been made to model each Wi-Fi user's channel access process as Markov chain \cite{bianchi2000performance}.
Works in \cite{chang2018traffic,kai2019energy} show that this model \cite{bianchi2000performance} can be used to formulate the optimization problem of user grouping problem in RAW, e.g., to balance the channel utilization \cite{chang2018traffic} and the energy efficiency \cite{kai2019energy} for user groups.
Note that the Markov chain in \cite{bianchi2000performance} only models the contention of users without considering the inter-user interference caused by hidden users in the network.
The model developed in \cite{garetto2008modeling} can also be used, which extends the work in \cite{bianchi2000performance} by modeling the interference from hidden users.
Thus, it achieves more accurate throughput estimation and can better capture the user throughput starvation effect in the network when applied to the user grouping problem.
However, applying this extended model \cite{garetto2008modeling} requires prior knowledge of all hidden-user pairs whose measurements introduce significant overheads in the network. The overheads limit the performance of Wi-Fi HaLow networks since the networks' radio resources, e.g., bandwidth and power, are limited.

\subsubsection{Graph-based Approach}
Wireless network optimization problems \cite{chen2022energy,chang2009multicell,liang2018graph} can be formulated as graph theory problems \cite{west2001introduction}, where the graph's vertices and edges represent the users and the contention or interference between the users, respectively.
The work in \cite{chen2022energy} constructs an unweighted undirected graph by connecting two users if an AP can detect both users' transmissions. Then, a graph coloring problem assigns time slots as colors (or groups) to users such that no adjacent users occupy the same slot.
Authors in \cite{chang2009multicell,liang2018graph} construct weighted graphs with edges indicating interference levels in user pairs. Here, undirected edges \cite{chang2009multicell} assume the interference equally impacts both users' throughput, while directed edges \cite{liang2018graph} differentiate the interference power made by two users.
Then, the graph's max cut divides users into a given number of groups by maximizing the sum of edge weights between groups, and each group is assigned orthogonal frequencies to manage the interference. 
Note that the graph constructions in \cite{chen2022energy,chang2009multicell,liang2018graph} heuristically use fixed rules predefined according to human intuition on how contention and interference affect the network performance.
Whether or not these graph constructions are optimal for the user grouping problem in RAW is unknown, and how to flexibly optimize the graph construction in wireless networks requires investigation.

\subsubsection{Machine Learning Approach}
Machine learning (ML) methods have recently been widely applied in wireless network optimizations \cite{szott2022wi}.
Typically, ML methods use network data to train a neural network (NN) that returns the optimal network control decisions for given network states \cite{gu2021knowledge}. Thus, they require no explicit model of the network behaviors such as contention and interference.
Note that classic fully connected NNs (FNNs) have pre-determined input and output dimensions, which are not flexible to varying dimensions of the problems, e.g., due to varying numbers of users in the network.
Instead, graph NNs (GNN) can be used with flexibility \cite{wu2020comprehensive}. 
This is because GNNs take network states as features on a graph whose size is adaptive to the dimensions of the problem. 
Then, GNNs obtain the control decisions by repeatedly aggregating features of neighboring vertices (or edges) along with trainable parameters \cite{eisen2020optimal,shen2020graph}. 
Nevertheless, research has not been done on applying ML-based approaches and designing NN structures to solve the user group problem in RAW.

\subsection{Our Contributions}
This paper studies scheduling RAW slots in Wi-Fi HaLow networks to avoid user throughput starvation. This is achieved by finding optimal user grouping decisions in RAW that maximize the worst-case user throughput \cite{huang2001max}.
We use path losses from users to APs as the network states to make the decisions, which are measurable at APs using signals sent by users \cite{pavon2003link} without additional measurement overheads.
First, we construct a fully connected weighted directed graph (or simply the graph onwards) to represent the network. 
Here, each user is a vertex. Each directed edge represents the asymmetric contention and interference between a user pair, and the edge weight indicates how negatively one user's transmissions impact the other user's throughput.
We apply the graph's max cut for the given edge weights to group users. Then, we propose adjusting the edge weights to construct the optimal graph whose max cut results in the optimal decisions.
Next, we design an ML algorithm that trains a graph-constructing actor (a FNN) to approximate optimal mapping from each user pair's states to the edge weights between them.
The algorithm contains a graph cut procedure using semidefinite programming (SDP) \cite{goemans1995improved} to efficiently solve the max cut of the given actor-constructed graph, which returns the grouping decisions.
Also, a graph-evaluating critic (a GNN) is trained to estimate user throughput for the given graph and states when using the above decisions, whose gradient optimizes the actor.
Here, we design the actor and the critic's structures so they can abstract contention and interference information from states, e.g., a part of NNs is optimized to infer probabilities of hidden user pairs.
Furthermore, since real-world networks can differ from those in offline training, we study how to adjust the graph based on online measurements.

The contributions of this paper are listed as follows.

\begin{itemize}
\item 
To the best of our knowledge, we are the first to propose an optimization framework that formulates the user grouping problem in RAW as a graph construction problem, where the graph's edge weights are optimization variables and the grouping decisions are computed from the constructed graph's max cut. Unlike the existing graph-based approach \cite{chen2022energy,chang2009multicell,liang2018graph} using heuristically constructed edges, the framework can flexibly optimize edge weights in the network's graph representation according to the specific network performance objective, i.e., maximizing the worst-case user throughput.

\item 
We are also the first to develop an actor-critic graph representation learning (AC-GRL) algorithm that trains NNs to construct the optimal graph representing the impact of interference in a wireless network. The interference in each individual user pair is represented as the edge weight in the pair. Note that the edge weight in the NN-constructed graph indicates how likely a pair of users belong to the same group when the graph's max cut is applied to obtain user grouping decisions. This is unlike the existing ML-based approach \cite{eisen2020optimal,shen2020graph} that directly uses NNs to generate user grouping decisions, which loses the above individual user-pairwise interference information by aggregating all neighboring users' features. Consequently, these methods fail to return useful decisions in user grouping.

\item 
We design a fine-tuning architecture that adjusts the NN-constructed graph according to user throughput and path losses measured online. 
Specifically, we first initialize the actor and critic using offline-trained NN parameters. The offline-trained actor then generates the graph's edge weights according to the initial network states. 
Then, the offline-trained critic continuously adjusts the graph's edge weights based on the user index with the worst-case throughput periodically measured from the network. 
At the same time, the actor also proportionally updates the weights based on the latest measured states.
After each edge weight update, user grouping decisions are recomputed by using the graph cut procedure.

\item We implement the proposed methods in a system-level simulation platform \cite{ns3}, NS-3, compliant with Wi-Fi standards, where we also study and apply the existing Markov-model-based \cite{chang2018traffic,kai2019energy}, graph-based \cite{chen2022energy,chang2009multicell,liang2018graph} and ML-based  \cite{eisen2020optimal,shen2020graph} approaches to the user grouping problem in RAW. Simulations show that our grouping decisions achieve around $30\%\sim80\%$ higher worst-case user throughput than the existing approaches.
Also, the proposed architecture can further improve the worst-case user throughput $5\%\sim30\%$ by online fine-tuning the graph constructed by the offline-trained NNs when users' locations are static. When users move, the proposed architecture achieves around $200\%$ higher worst-case user throughput by online fine-tuning the graph over time, compared to the fixed graph as its initial construction.
\end{itemize}

\subsection{Notation} 
The $i$-th element of a vector $\mathbf{x}$ is denoted as $x_{i}$.
The $j$-th element of the $i$-th row of a matrix $\mathbf{X}$ is denoted as $X_{i,j}$.
We write the definition of $\mathbf{X}$'s elements as $\mathbf{X}\triangleq[X_{i,j}|X_{i,j} = (\cdots)]$, where $(\cdots)$ is the expression defining $\mathbf{X}$'s elements.
$\diag\{\mathbf{X}\}$ denotes $\mathbf{X}$'s diagonal elements.
We write $\mathbf{X} \succeq 0$ to indicate that $\mathbf{X}$ is positive semidefinite. $\svd(\mathbf{X})$ is a singular value decomposition of $\mathbf{X}$. 
$\sgn(\mathbf{x})$ is a vector of each element's sign ($+1$ or $-1$) in $\mathbf{x}$. $\mathbf{1}_{\{\cdots\}}$ is an indicator function that equals to $1$ if the expression $\{\cdots\}$ is true or otherwise $0$.

\section{System Model and Problem Formulation}\label{sec:wifi_system_model}
This section first presents the system model of the wireless network, including the network configurations and states. Then, we formulate the user grouping problem for contention and interference management.

\subsection{Configurations of the Wireless Network}\label{subsec:wifi_net_configuration}
\begin{figure}[!ht]
\centering
\includegraphics[scale=0.65]{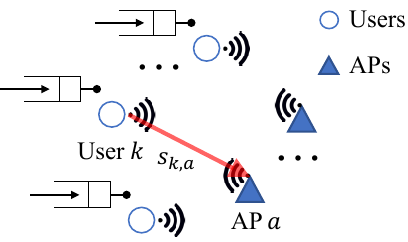}
\caption{Illustration of a wireless network.}
\label{fig:wifi_met}
\vspace{-0.2cm}
\end{figure}
We consider a random-access-based wireless network, e.g., a Wi-Fi 802.11ah (HaLow) network, that consists of $K$ users and $A$ APs, as illustrated in Fig. \ref{fig:wifi_met}. 
Here, we assume that users are sensors continuously collecting data to monitor critical information (e.g., on healthcare systems or environmental hazards). Users return the collected data samples to APs in the uplink \cite{aust2015outdoor}.
We assume that all users and APs operate in the same channel with a bandwidth $B$ in Hertz (Hz).
We denote the transmission power of users as $\mathbf{P}_\mathrm{0}$ in dBm and the noise power spectral density as $\mathbb{N}_\mathrm{0}$ in dBm/Hz.
We denote the path loss from the $k$-th user to the $a$-th AP as $\Tilde{s}_{k,a}$ in dB, $\forall k, a$. We assume each user is associated with and transmits to the AP with the minimum path loss to the user. The AP that the $k$-th user is associated with is denoted as $\hat{a}(k)$, e.g.,
\begin{equation}\label{eq:associated_ap_id}
\begin{aligned}
\hat{a}(k) \triangleq \arg\min_a \Tilde{s}_{k,a} \ ,\ \forall k \ .
\end{aligned}
\end{equation}
We assume that all sensors (or users) employ the same data-collecting mechanism that remains unchanged over time and that each collected data sample is encapsulated in a specific format with a constant length. Thus, we can assume that each user has a stationary packet arrival process with the packet size $L$ in bits.
Packets arrived at each user are stored in a first-in-first-out queue with a constant queue size, where the oldest packet is discarded when a new packet arrives and the queue is full.
The duration to transmit each $L$-bit packet of user $k$, $d_k$ in seconds, depends on the modulation and coding scheme (MCS).
We configure each user's MCS to satisfy that the decoding error probability without interference, $\epsilon_k$, is lower than a threshold $\epsilon_{\max}$.
Here, $\epsilon_k$ can be approximated as \cite{yang2014quasi}
\begin{equation}\label{eq:tx_error_for_mcs_wi-fi}
\begin{aligned}
 \epsilon_k \approx f_Q\Bigg(\frac{  -L\ln{2}+ {d_k B}\ln(1+\phi)}{\sqrt{d_k B V}}\Bigg)\  ,\ \forall k ,
\end{aligned}
\end{equation}
where $\phi = (\mathbf{P}_\mathrm{0}/\Tilde{s}_{k,\hat{a}(k)})/(\mathbb{N}_\mathrm{0}B)$ is the signal-to-noise-ratio (SNR) of user $k$.
Also, $f_Q$ is the tail distribution function of the standard normal distribution, and $V$ is the channel dispersion defined as $V = 1- {1}/{\big[1+\phi\big]^2}$ \cite{yang2014quasi} in \eqref{eq:tx_error_for_mcs_wi-fi}.
Here, since $\mathbf{P}_\mathrm{0}$, $\mathbb{N}_\mathrm{0}$, $B$ and $L$ in \eqref{eq:tx_error_for_mcs_wi-fi} are all assumed as constants, the packet duration of each user only depends on the path loss to its associated AP, $\Tilde{s}_{k,\hat{a}(k)}$, $\forall k$.

Each user operates a CSMA/CA process to access the channel and transmit their packets \cite{80211}, avoiding collisions with other users' packet transmissions. Specifically, each user first senses whether or not any other user is transmitting in the channel. 
If the above is true, the user continues to monitor the channel until the channel is sensed as idle. 
Then, the user backs off for a random time, which reduces the probability of collision with other users also waiting for transmission opportunities.
If the channel is sensed as busy again during the backoff, the backoff countdown is paused and resumed again after the channel is sensed as idle.
After the backoff ends, the user transmits the oldest packet from its queue over the channel.
If some users cannot sense this transmission, meaning they are hidden from that user, they will make simultaneous transmissions. 
These transmissions cause interference at the receiving AP and raise the decoding error probability.
When a packet transmission is failed, a negative acknowledgment is generated, and the user retries transmitting the same packet. Retransmissions continue until the user receives an acknowledgment or reaches the maximum attempts. The user then transmits the next packet in the queue using the same process.

\subsection{Network States}\label{subsec:wifi_net_states}
We assume that if the path loss between a user and an AP is greater than a threshold $\Tilde{s}_{\max}$ in dB, then the user's transmissions cannot be detected and received at the AP. Thus, this path loss cannot be measured from the network. 
We write the measured path loss from the $k$-th user to the $a$-th AP as
\begin{equation}
\begin{aligned}
s_{k,a} \triangleq 
\begin{cases}
    \Tilde{s}_{k,a}\ ,\ \text{if } \Tilde{s}_{k,a}\leq \Tilde{s}_{\max} \ , \\
    2\Tilde{s}_{\max}\ ,\ \text{if } \Tilde{s}_{k,a}>\Tilde{s}_{\max}  \ ,
\end{cases} \forall k,a \ ,
\end{aligned}
\end{equation}
where if $\Tilde{s}_{s,a}> \Tilde{s}_{\max}$, we set the value of $s_{k,a}$ as $2\Tilde{s}_{\max}$, indicating that it is immeasurable. 
Note that the path loss to each user's associated AP is always measurable because it is the user's minimum path loss to APs (otherwise, the user is not connected in the network), i.e.,  $\Tilde{s}_{k,\hat{a}(k)}=s_{k,\hat{a}(k)},\ \forall k$.
Each user's states are its measured path losses to APs, defined as\footnote{The path losses are normalized based on the maximum path loss where users' signals are detectable at APs as $s_{k,a}\leftarrow s_{k,a}/\Tilde{s}_{\max}-1$, resulting in path losses' range as $[-1,1]$. For simplicity of notations, we assume that path losses $s_{k,a}$, $\forall k,a$, have been normalized onwards unless specifically stated.}
\begin{equation}
\begin{aligned}
\mathbf{s}_k \triangleq [s_{k,1},\dots,s_{k,A}]^\textrm{T} \ ,\ \forall k \ ,
\end{aligned}
\end{equation}
and the whole network's states is a $A\times K$ matrix that contains all users' measured path losses, defined as
\begin{equation}
\begin{aligned}
\mathbf{S} \triangleq [\mathbf{s}_{1},\dots,\mathbf{s}_{K}] \ .
\end{aligned}
\end{equation}

\subsection{User Grouping Problem in RAW}\label{subsec:user_grouping_problem_formulation}
\begin{figure}[!ht]
\centering
\includegraphics[scale=0.65]{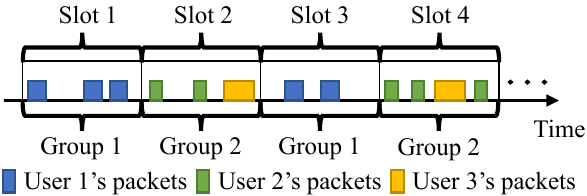}
\caption{Illustration of user grouping in RAW, where $K=3$, $Z=2$, $z_1=1$ and $z_2=z_3=2$.}
\label{fig:raw_slot_and_goup}
\vspace{-0.2cm}
\end{figure}
Fig. \ref{fig:raw_slot_and_goup} illustrates a user grouping scheme in RAW of the Wi-Fi HaLow networks.
We assume that all users and APs are synchronized in time, and the time is slotted in RAW with indices $1,\dots,t,\dots$. 
The duration of each RAW slot is denoted as $\Delta_\mathrm{0}$ in seconds. 
We manage the contention and interference in the network by dividing users into $Z$ groups (for simplicity, we assume that $Z$ is an exponent of $2$).
Then, different groups of users transmit their packets in separate periodical RAW slots.
Specifically, We denote the group of user $1,\dots,K$ as $\mathbf{z}\triangleq[z_1,\dots,z_K]^{\rm T}$, where $z_k$ is user $k$'s group satisfying
\begin{equation}\label{eq:const:z_k_in_Z}
\begin{aligned}
z_k \in\{1,\dots,Z\}\ ,\ \forall k \ ,
\end{aligned}
\end{equation}
Users in $z$-th group are $\{k|z_k=z\}$, and they transmit only in periodical time slots $t=z,z+Z,z+2Z,\dots$.
As a result, contention and interference only happen among users within the same group and are eliminated between any two groups.
For example, in Fig. \ref{fig:raw_slot_and_goup}, three users ($K=3$) are grouped into two groups in RAW ($Z=2$), with user $1$ in group $1$ and user $2$ and $3$ in group $2$. Hence, group $1$ (user $1$) and group $2$ (user $2$ and $3$) transmit in slots $1,3,5\dots$ and $2,4,6\dots$, respectively. This user grouping scheme eliminates the contention and interference between user $1$ and $2$ and between user $1$ and $3$.
We use $u_k(t)$ to denote the number of packets successfully transmitted by the $k$-th user within RAW slot $t$. 
Note that $u_k(t)=0$ if $t\neq z_k,z_k+Z,z_k+2Z,\dots$ because users cannot transmit packets in slots other than the slot assigned to them. 
We denote the throughput of user $1,\dots,K$ as $\mathbf{r}\triangleq[r_1,\dots,r_K]^{\rm T}$, where $r_k$ is defined as
\begin{equation}\label{eq:asymptotic_average_user_throughput}
\begin{aligned}
r_k
&\triangleq \lim_{T\to \infty} \frac{1}{T} \sum_{t=1}^{T} u_k(t) \ , \forall k ,
\end{aligned}
\end{equation}
which is the average number of packets transmitted over time.  

We write $\expt[r_k|\mathbf{S},\mathbf{z}]$, $\forall k$, as the expected throughput of user $k$ for given network states $\mathbf{S}$ and grouping decisions $\mathbf{z}$.
We aim to maximize the network performance objective as the worst-case user throughput by controlling $\mathbf{z}$ as
\begin{equation}\label{eq:prob:user_grouping}
\begin{aligned}
\max_{\mathbf{z}}     \min_{k} \expt[r_k|\mathbf{S},\mathbf{z} ] \ , \  \text{s.t. } \eqref{eq:const:z_k_in_Z}\ ,
\end{aligned}
\end{equation}
where $\min_k \expt[r_k|\mathbf{S},\mathbf{z} ]$ is the worst-case user throughput for given $\mathbf{S}$ and $\mathbf{z}$. 
Note that we assume that users' locations are static in the above formulation, where the network states $\mathbf{S}$ (or the user-to-AP path losses) do not change over time. Therefore, we can asymptotically average the user throughput $\mathbf{r}$ in \eqref{eq:asymptotic_average_user_throughput} for given grouping decisions $\mathbf{z}$.
We will study our methods for solving the user grouping problem under this assumption in Sections \ref{sec:graph_representation_framework} and \ref{sec:ac_grl}, where $\mathbf{z}$ is decided at the initialization of the static network and remains the same over time. Further, we will design an online architecture in Section \ref{sec:online_arch} to periodically update $\mathbf{z}$ according to online measured user throughput and path losses that may change over time due to user mobility.

The straightforward problem formulation of user grouping in \eqref{eq:prob:user_grouping} is challenging to solve because no clear structure shows how to map the given network states $\mathbf{S}$ and the grouping decisions $\mathbf{z}$ to the user throughput $\mathbf{r}$. The models \cite{bianchi2000performance,garetto2008modeling} can establish the mapping in $\mathbf{S}$, $\mathbf{z}$ and $\mathbf{r}$ \cite{chang2018traffic,chen2022energy}. However, this approach requires full information on the path losses between each user pair to determine contending/hidden users whose measurements introduce substantial overheads in the network. Moreover, even if there is full information on contending/hidden users, these models are non-linear and have no closed-form expressions \cite{bianchi2000performance,garetto2008modeling}, which can hardly be applied for optimizing user grouping decisions. Therefore, we propose a new framework for user grouping below.

\section{Graph Representation Optimization for User Grouping in Wireless Networks}\label{sec:graph_representation_framework}
This section presents the framework that formulates the user grouping problem as the graph construction problem that optimizes a graph to represent the wireless network. 
Since the contention and interference are user-pairwise interactions in the network, they can be modeled as an edge connection between a user pair to represent their effects on the network performance \cite{chen2022energy,chang2009multicell,liang2018graph}, e.g., throughput. Each edge can be designed based on user-pairwise network states, i.e., $\mathbf{s}_i$ and $\mathbf{s}_j$, $i\neq j$ and users can be grouped by cutting some of the edges between users. By using a specific edge-cut mechanism (e.g., the max-cut), we only need to optimize the representation of the edges according to the user-pairwise states. This creates a structure that links network states with the edges and further with the grouping decisions, where we can eventually control the decisions by optimizing the edge design.

We model the network as a fully connected directed graph as $\mathcal{G} = (\mathcal{V},\mathcal{E})$, where $\mathcal{V}$ and $\mathcal{E}$ are the set of vertices representing users and the set of edges representing the contention and interference between users, respectively. They are defined as
\begin{equation}
\begin{aligned}
\mathcal{V}\triangleq\{1,\dots,K\} \ ,\  \mathcal{E}\triangleq\{(i,j)|\forall i,j \in \mathcal{V}, \ i\neq j\}\ ,
\end{aligned}
\end{equation}
where $(i,j)$ is the edge from the $i$-th to the $j$-th user.
The edge direction from $i$ to $j$ expresses the contention and interference that user $i$'s transmissions cause to user $j$. 
Note that the impacts of interference and contention are different, e.g., decreasing the SINR versus increasing the backoff time.
However, since contending/hidden users are not measured in the network, we cannot explicitly differentiate the impact of contention and interference. Also, as there is only one performance metric, i.e., throughput, we aggregate interference and contention's impact on throughput as one scalar on each edge.
Specifically, each edge has a weight $W_{i,j}$ representing how the contention and interference caused by user $i$ negatively impact user $j$'s throughput.
We assume that each edge weight is a bounded non-negative real number, i.e., $W_{i,j}\in[0,1]$ (Without loss of generality, we assume that the upper bound of weights is $1$), $\forall i,j=1,\dots,K$ and $i\neq j$.
Here, the larger $W_{i,j}$ implies that user $j$'s throughput is more negatively affected by user $i$'s transmissions.
Note that $W_{i,j}$ is generally not equal to $W_{j,i}$ due to differences in users' packet duration and path losses to APs.
We collect all edge weights in the graph's adjacency matrix as
\begin{equation}\label{eq:defi:graph_adjacency_matrix}
\begin{aligned}
\mathbf{W} \triangleq \big[W_{i,j}\big|W_{i,j}\in[0,1],\forall i\neq j;W_{k,k}=0,\forall k \big] \ .
\end{aligned}
\end{equation}
where $\mathbf{W}$ is a $K\times K$ matrix whose diagonal, $W_{k,k}$, $\forall k$, is equal to $0$, i.e., the graph has no self-loop edge on vertices.

We can obtain the user grouping decisions by the max cut on the above graph as 
\begin{equation}\label{eq:prob:user_grouping:max_cut}
\begin{aligned}
\max_{\mathbf{z}} \sum_{i\neq j} W_{i,j} \mathbf{1}_{\{z_i\neq z_j\}}, \  \text{s.t. }\ \eqref{eq:const:z_k_in_Z}\ ,
\end{aligned}
\end{equation}
which maximizes the contention and interference, represented by edge weights, eliminated between different groups. Here, $\mathbf{1}_{\{z_i\neq z_j\}}$ in \eqref{eq:prob:user_grouping:max_cut} is the indicator function that equals to $1$ if user $i$ and $j$ are assigned to different groups or otherwise equals to $0$. In the above max cut problem, the higher the edge weights between two users, the more likely they are divided into two groups, leading to minimum contention and interference within each group. 
Based on this fact, we reformulate the user grouping problem in \eqref{eq:prob:user_grouping} as a bi-level optimization problem as
\begin{equation}\label{eq:prob:user_grouping:graph_cut:adaptive_edge_weighting}
\begin{aligned}
\max_{W_{i,j},\forall i\neq j}    &\ \min_k \expt[r_k|\mathbf{S},\mathbf{z}]\ ,      \\ 
                \textrm{s.t.}   &\ \  W_{i,j}\in[0,1] , \  \forall i\neq j \ , \\  
                                &\ \ \mathbf{z}=\arg\max_{\mathbf{z}' }\  \sum_{i\neq j} W_{i,j} \mathbf{1}_{\{z'_i\neq z'_j\}} , \ \text{s.t. } \eqref{eq:const:z_k_in_Z} .
\end{aligned}
\end{equation}
where the max cut problem in \eqref{eq:prob:user_grouping:max_cut} is embedded as the lower-level problem (LLP), and $\mathbf{z}$ is decided as the solution of the LLP.
We can show that
\begin{lemma}\label{lemma:equivalence_of_user_grouping_problem}
Define the optimal edge weights as $W^*_{i,j}$, $\forall i\neq j$, that maximize the network performance objective in \eqref{eq:prob:user_grouping:graph_cut:adaptive_edge_weighting}, and define
$\mathbf{z}^*$ as the optimal solution of the LLP of \eqref{eq:prob:user_grouping:graph_cut:adaptive_edge_weighting} for the above optimal edge weights. Then, $\mathbf{z}^*$ also maximizes the objective in \eqref{eq:prob:user_grouping}.
\end{lemma}
\begin{proof}
The proof is in the appendix.
\end{proof}
The above statement implies that we can find the optimal grouping decisions in \eqref{eq:prob:user_grouping} by solving the problem \eqref{eq:prob:user_grouping:graph_cut:adaptive_edge_weighting} instead. 
Also, note that the grouping decisions in \eqref{eq:prob:user_grouping:graph_cut:adaptive_edge_weighting} only depend on the edge weights, as shown in the LLP of \eqref{eq:prob:user_grouping:graph_cut:adaptive_edge_weighting}.
As a result, the formulation in this section transforms the user grouping problem into a problem that finds the optimal edge weights in the graph representing the network. 
Furthermore, we assume that each of the optimal edge weights in the graph, $W^*_{i,j}$, is a function of user $i$ and user $j$' states, $\mathbf{s}_i$ and $\mathbf{s}_j$, as
\begin{equation}\label{eq:assu:user_grouping:w_equal_to_mu_si_sj}
\begin{aligned}
W^*_{i,j} = \mu^*(\mathbf{s}_i,\mathbf{s}_j) \ , \forall i\neq j \ ,
\end{aligned}
\end{equation}
where $\mu^*(\cdot)$ is referred to as the optimal graph-constructing function. It maps each user pair's states to the optimal edge weight between them. 

Based on the formulation in \eqref{eq:prob:user_grouping:graph_cut:adaptive_edge_weighting} and the assumption in \eqref{eq:assu:user_grouping:w_equal_to_mu_si_sj}, the user grouping problem in \eqref{eq:prob:user_grouping} can be solved by finding the optimal graph-constructing function that uses user-pairwise states to generate the graph's edge weights. 
We note that the existing graph-based methods \cite{chen2022energy,chang2009multicell,liang2018graph} provide some manual designs of the graph-constructing function in the proposed framework.
For example, work in \cite{chen2022energy} designs the function as $\mu(\mathbf{s}_i,\mathbf{s}_j)=\mathbf{1}_{\{\hat{a}(i)= \hat{a}(j)\}}, \forall i,j$, to indicate whether two users are associated with the same AP ($\hat{a}(k),\forall k$ is computed using $\mathbf{s}_i$ and $\mathbf{s}_j$ as \eqref{eq:associated_ap_id}), or work in \cite{liang2018graph} designs the function as $\mu(\mathbf{s}_i,\mathbf{s}_j)=s_{i,\hat{a}(j)}, \forall i,j$, to set the edge as the path loss from one user to another user's associated AP. However, because these heuristic methods use a fixed function, their output edge weights cannot be adjusted to reflect the contention and interference's impact on a given performance objective, e.g., avoiding throughput starvation.
To address this issue, we propose approximating the optimal graph-constructing function as a NN whose parameters can be flexibly trained by ML. We study the design of graph-constructing NN and its training algorithm in the next.

\section{Actor-Critic Graph Representation Learning}\label{sec:ac_grl}
\begin{figure}[!ht]
\centering
\includegraphics[scale=0.65]{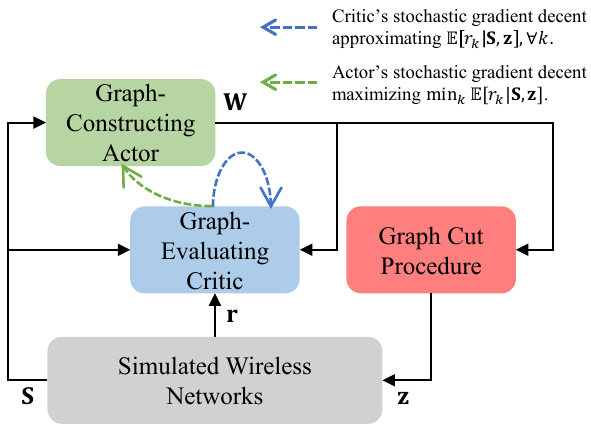}
\caption{Overall structure of the AC-GRL algorithm.}
\label{fig:actor_critic_algorithm}
\vspace{-0.2cm}
\end{figure}

In this section, we develop the AC-GRL algorithm that trains the NN, referred to as the graph-constructing actor, to construct optimal edge weights in the graph \cite{karasuyama2017adaptive}. Note that optimal edge weights maximize the worst-case user throughput in \eqref{eq:prob:user_grouping:graph_cut:adaptive_edge_weighting} for given network states.
Training the actor requires the gradients of the actor's output with respect to the specific objective value. However, traditional model-based approaches cannot map the actor's edge weight outputs (or the grouping decisions) to the user throughput in the objective of \eqref{eq:prob:user_grouping:graph_cut:adaptive_edge_weighting} because of unmeasured contending/hidden users and the lack of closed-form expressions, as mentioned in Section \ref{sec:wifi_system_model}. This issue is solved by using an additional NN, referred to as the graph-evaluating critic, to approximate the user throughput for the graph's edge weights and network states. Consequently, our NN training algorithm has an actor-critic form.

Fig. \ref{fig:actor_critic_algorithm} shows the structure of the AC-GRL algorithm with 1) the graph-constructing actor that uses network states $\mathbf{S}$ to infer probabilities of hidden users and generate the graph's edge weights collected in the adjacency matrix $\mathbf{W}$. 2) the graph cut procedure that performs the graph's max cut to generate the user grouping decisions $\mathbf{z}$ according to the LLP of \eqref{eq:prob:user_grouping:graph_cut:adaptive_edge_weighting} with given $\mathbf{W}$, and 3) the graph-evaluating critic that performs the same hidden user inference as the actor and further evaluates how good is the graph for the given $\mathbf{W}$ and $\mathbf{S}$ based on measured user throughput $\mathbf{r}$. The remaining section first presents the design of the above three components in Sections \ref{subsec:actor_nn_structure}, \ref{subsec:GM_graph_cut} and \ref{subsec:critic_nn_structure} and then explains the flow of the AC-GRL algorithm in Section \ref{subsec:grl_algorithm_flow}.

\subsection{Design of the Graph-Constructing Actor}\label{subsec:actor_nn_structure}
\begin{figure}[!ht]
\centering
\includegraphics[scale=0.65]{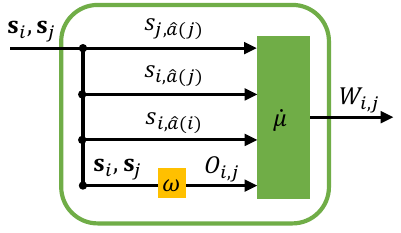}
\caption{The structure of the actor.}
\label{fig:actor_nn_arch}
\vspace{-0.2cm}
\end{figure}

We use a NN $\mu(\cdot|\theta^\mu)$ with trainable parameters $\theta^\mu$ to approximate $\mu^*(\cdot)$ in \eqref{eq:assu:user_grouping:w_equal_to_mu_si_sj}, i.e., the optimal mapping from the users' states $\mathbf{S}=[\mathbf{s}_{1},\dots,\mathbf{s}_{K}]$ to optimal edge weights, as
\begin{equation}\label{eq:actor_definition}
\begin{aligned}
W_{i,j} \triangleq \mu(\mathbf{s}_i,\mathbf{s}_j|\theta^\mu) \approx \mu^*(\mathbf{s}_i,\mathbf{s}_j) \ , \forall i \neq j \ .
\end{aligned}
\end{equation}
where $\mu(\cdot|\theta^\mu)$ is the graph-constructing actor. In the actor, we pre-process the network states and extract key information on contention and interference to help the actor decide the edge weight values. Specifically, for a given pair of users, $i$ and $j$, $i\neq j$, we train the actor to set $W_{i,j}$' value to differentiate how negatively user $i$'s transmissions affect user $j$'s throughput according to the following information in $\mathbf{s}_i$ and $\mathbf{s}_j$: 
\begin{itemize}
    \item $s_{j,\hat{a}(j)}$, the path loss from user $j$ to its associated AP $\hat{a}(j)$. This feature determines the receiving signal power of user $j$ at the AP. 
    As previously mentioned, it also determines user $j$'s packet duration. 
    A smaller value of this feature corresponds to a shorter packet duration and a stronger signal strength for user $j$. 
    This means that user $j$'s transmissions are more robust to contention and interference. Thus, the well-trained actor should set $W_{i,j}$ to a relatively smaller value when $s_{j,\hat{a}(j)}$ is larger.
    \item $s_{i,\hat{a}(j)}$, the path loss from user $i$ to user $j$'s associated AP $\hat{a}(j)$, This value indicates the interference power at the AP. 
    The smaller $s_{i,\hat{a}(j)}$ suggests that $W_{i,j}$ should be set to a higher value by the well-trained actor, accounting for larger interference strength.
    \item $s_{i,\hat{a}(i)}$, the path loss from user $i$ to its associated AP $\hat{a}(i)$, which determines the packets' duration sent by user $i$, i.e., the interference duration caused by user $i$. 
    The smaller $s_{i,\hat{a}(i)}$ indicates that user $i$'s transmissions have a shorter duration, resulting in a shorter interference or contention duration to user $j$ when they are hidden or contending users, respectively. 
    Thus, the smaller $s_{i,\hat{a}(i)}$ suggests the well-trained actor should set $W_{i,j}$ to a smaller value.   
    \item Whether user $j$ can sense user $i$'s transmissions, i.e., whether user $i$ is contending with or hidden from user $j$. Since this information is not measured in the network states, we use a FNN $\omega(\cdot|\theta^\omega)$ with parameters $\theta^\omega$ to infer it from the measured network states as
    \begin{equation}\label{eq:infer_hidden_target}
    \begin{aligned}
    &O_{i,j} 
    \triangleq \omega(\mathbf{s}_i, \mathbf{s}_j|\theta^\omega)\\
    &\approx \Pr\big[\text{user $j$ senses user $i$'s transmissions}\big| \mathbf{s}_i, \mathbf{s}_j\big] \ ,
    \end{aligned}
    \end{equation}
    where $1-O_{i,j}$ indicates how likely user $i$ is hidden from user $j$ for given $\mathbf{s}_i$ and $\mathbf{s}_j$. We refer to $\omega(\cdot|\theta^\omega)$ as the inference NN. $O_{i,j}$ informs the actor how likely user $i$ and $j$ are contending or hidden users, and the actor should combine $O_{i,j}$ with the previously extracted information, i.e., $s_{j,\hat{a}(j)}$, $s_{i,\hat{a}(j)}$ and $s_{i,\hat{a}(i)}$, when deciding $W_{i,j}$.  
\end{itemize}
Based on the above pre-processing, we design the structure of the actor in \eqref{eq:actor_definition} based on the extracted contention and interference information as 
\begin{equation}\label{eq:actor_nn_structure}
\begin{aligned}
&W_{i,j} 
= \mu(\mathbf{s}_i,\mathbf{s}_j|\theta^\mu) = \dot{\mu}(s_{j,\hat{a}(j)},s_{i,\hat{a}(j)},s_{i,\hat{a}(i)},O_{i,j}|\theta^{\dot{\mu} }) \\
= &\dot{\mu}\big(s_{j,\hat{a}(j)},s_{i,\hat{a}(j)},s_{i,\hat{a}(i)},\omega(\mathbf{s}_i, \mathbf{s}_j|\theta^\omega)\big|\theta^{\dot{\mu} }\big) \  , \ \forall i\neq j\ ,
\end{aligned}
\end{equation}
where $\omega(\cdot|\theta^\omega)$ is considered as a part of the actor and $\dot{\mu}(\cdot|\theta^{\dot{\mu}})$ is also designed as a FNN. The above actor-generated edge weights are collected and returned as the graph's adjacency matrix $\mathbf{W}$ (with all $0$ diagonal elements, as defined in \eqref{eq:defi:graph_adjacency_matrix}).

We discuss how edge weight values affect the grouping decisions generated from the graph's max cut. Clearly, if two users are heavily interfered or contending with each other, larger edge weights will be set between them and the users are more likely separated into different groups. Meanwhile, if a user receives weak interference from (or barely contends with) its neighboring users, small edge weights will be assigned between them. In this case, this user is more likely separated from its neighboring users when the sum of its edge weights is large (implying a significant aggregated interference or contention). Also, when all edge weights are small values, e.g., in sparse networks with little interference or contention, the graph's max cut will balance the number of users in each group by cutting the edges with relatively larger weights.

\subsection{Design of the Graph Cut Procedure}\label{subsec:GM_graph_cut}
\begin{figure}[!ht]
\centering
\includegraphics[scale=0.65]{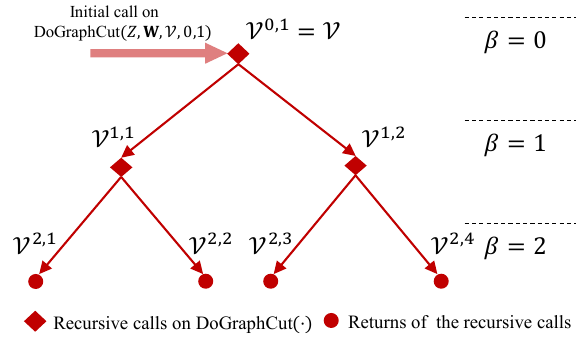}
\caption{Tree diagram illustrating the recursive graph cut when $Z=4$ and $\beta=0,1,\dots,\log_2(Z)$.}
\label{fig:recur_cut}
\vspace{-0.2cm}
\end{figure}
The graph cut procedure solves the graph's max cut problem in \eqref{eq:prob:user_grouping:max_cut} for given actor-generated edge weights $\mathbf{W}$. Specifically, it cuts the graph $\mathcal{G}$ into $Z$ parts and equivalently groups the graph's vertices (or the users) $\mathcal{V}$ into $Z$ groups.
As $Z$ is assumed to be a power of $2$, as mentioned in Section \ref{subsec:user_grouping_problem_formulation}, we can cut the graph $\mathcal{G}$ recursively by first dividing $\mathcal{V}$ into two subsets and repeatedly dividing each subset into two until there are $Z$ subsets. The division of each (sub)set aims to disconnect the edges with higher weights by solving the max cut within the (sub)set.
Fig.~\ref{fig:recur_cut} illustrates a tree diagram of the graph cut process with $\log_2(Z)+1$ levels of graph cut. In the $\beta$-th level, the subsets of users are $\mathcal{V}^{\beta,c}$, where $c = 1,\dots,2^\beta$ and $\beta = 0,\dots,\log_2(Z)$. The users in $\mathcal{V}^{\beta,c}$ are
\begin{equation}
\begin{aligned}
\mathcal{V}^{\beta,c} \triangleq \{k^{\beta,c}_1,\dots,k^{\beta,c}_{|\mathcal{V}^{\beta,c}|}\} \ , \forall \beta, c \ ,
\end{aligned}
\end{equation}
which are further divided into two subsets, $\mathcal{V}^{\beta+1,2c-1}$ and $\mathcal{V}^{\beta+1,2c}$, in the next level.
The edge weights between vertices in $\mathcal{V}^{\beta,c}$ are collected as $\mathbf{W}^{\beta,c}$, a $|\mathcal{V}^{\beta,c}|\times|\mathcal{V}^{\beta,c}|$ matrix whose elements are determined by the original graph's adjacency matrix $\mathbf{W}$ as
\begin{equation}\label{eq:w_setup}
\begin{aligned}
\mathbf{W}^{\beta,c} \triangleq \big[W^{\beta,c}_{i,j}\big| W^{\beta,c}_{i,j}= W_{k^{\beta,c}_i,k^{\beta,c}_j} \ , \ \forall k^{\beta,c}_i, k^{\beta,c}_j \in  \mathcal{V}^{\beta,c}\big] \ .
\end{aligned}
\end{equation}

We use Goemans and Williamson's method \cite{goemans1995improved} to divide a given (sub)set of vertices.
In detail, let us define an indicator $y^{\beta,c}_i$ for each user $k^{\beta,c}_i$ in $\mathcal{V}^{\beta,c}$, where $y^{\beta,c}_i=-1$ if user $k^{\beta,c}_i$ is divided into $\mathcal{V}^{\beta+1,2c-1}$ and otherwise, $y^{\beta,c}_i=+1$, i.e., it is divided into $\mathcal{V}^{\beta+1,2c}$.
Then,  we can write the sub-max-cut-problem that maximizes the sum of weights on the disconnected edges between $\mathcal{V}^{\beta+1,2c-1}$ and $\mathcal{V}^{\beta+1,2c}$ as
\begin{equation}\label{eq:prob:sub_cut}
\begin{aligned}
\max_{\mathbf{y}^{\beta,c}}&\ \sum_{i\neq j} W^{\beta,c}_{i,j} \frac{1-y^{\beta,c}_iy^{\beta,c}_j}{2} \ ,\\
\text{s.t.}&\ y^{\beta,c}_i \in \{ -1, +1\}\ ,\ \forall i = 1,\dots,|\mathcal{V}^{\beta,c}|\ ,
\end{aligned}
\end{equation}
where $\mathbf{y}^{\beta,c}\triangleq[ y^{\beta,c}_{1},\dots,y^{\beta,c}_{|\mathcal{V}^{\beta,c}|} ]^{\mathrm{T}}$. Here, $(1-y^{\beta,c}_iy^{\beta,c}_j)/2$ in \eqref{eq:prob:sub_cut} is equal to $1$ if $k^{\beta,c}_i$ and $k^{\beta,c}_j$ are not in the same subset or $0$ otherwise, indicating whether or not edge $(k^{\beta,c}_i,k^{\beta,c}_j)$ is disconnected.
This problem in \eqref{eq:prob:sub_cut} can be relaxed into a SDP problem \cite{goemans1995improved} as
\begin{equation}\label{eq:prob:sub_cut_sdp}
\begin{aligned}
\hat{\mathbf{X}}^{\beta,c} = \arg\max_{\mathbf{X}^{\beta,c}}&\  \sum_{i\neq j} W^{\beta,c}_{i,j}\frac{1-X^{\beta,c}_{i,j}}{2} \ ,\\
\text{s.t.}&\ , \ \diag\{\mathbf{X}^{\beta,c}\} = 1\ ,\ \mathbf{X}^{\beta,c}\succeq0\ .
\end{aligned}
\end{equation}
Here, $\mathbf{X}^{\beta,c}$ in \eqref{eq:prob:sub_cut_sdp} is a $|\mathcal{V}^{\beta,c}|\times|\mathcal{V}^{\beta,c}|$ matrix whose elements $X^{\beta,c}_{i,j}$ are real numbers that approximate the multiplication $y^{\beta,c}_iy^{\beta,c}_j$ in \eqref{eq:prob:sub_cut}, $\forall i,j$.
The above problem can be solved by existing convex optimization solvers \cite{diamond2016cvxpy}.

After solving the problem in \eqref{eq:prob:sub_cut_sdp}, we can obtain the graph cut indicators $\mathbf{y}^{\beta,c}$ by a rounding method \cite{goemans1995improved}. Specifically, it performs singular value decomposition (SVD) on the optimal solution $\hat{\mathbf{X}}^{\beta,c}$ that maximizes the objective in \eqref{eq:prob:sub_cut_sdp}. Since $\hat{\mathbf{X}}^{\beta,c}$ is positive semidefinite, we can obtain a SVD of $\hat{\mathbf{X}}^{\beta,c}$ in the following structure \cite{strang2006linear},
\begin{equation}\label{eq:svd}
\begin{aligned}
\mathbf{U}^{\beta,c}\mathbf{\Sigma}^{\beta,c} (\mathbf{U}^{\beta,c})^\mathrm{T} = \svd(\hat{\mathbf{X}}^{\beta,c}) \ ,
\end{aligned}
\end{equation}
where $ \mathbf{U}^{\beta,c}$ and $ \mathbf{\Sigma}^{\beta,c}$ are both $|\mathcal{V}^{\beta,c}|\times |\mathcal{V}^{\beta,c}|$ matrices and further $\mathbf{\Sigma}^{\beta,c}$ is a diagonal matrix with non-negative eigenvalues of $\hat{\mathbf{X}}^{\beta,c}$ on its diagonal.
Then, we randomly select a vector in $\mathbb{R}^{|\mathcal{V}^{\beta,c}|}$ as $\delta^{\beta,c}$ and set $\mathbf{y}^{\beta,c}$ as\footnote{Each row of $\mathbf{U}^{\beta,c} (\mathbf{\Sigma}^{\beta,c})^{\frac{1}{2}}$ is a unit vector, and the size of the angle between the $i$-th and $j$-th row indicates how likely user $i$ and $j$ should be assigned in the same subset, $\forall i\neq j$ (smaller the angle, more likely they are in the same subset). Then, by multiplying $\mathbf{U}^{\beta,c} (\mathbf{\Sigma}^{\beta,c})^{\frac{1}{2}}$ with a random vector, users' grouping indicators generated by \eqref{eq:get_y_by_sgn} are more likely to have the same sign if the angle between their corresponding row vectors is smaller.
It has been proved \cite{goemans1995improved} that this method achieves at least $0.87854$ of optimal achievable max cut objective in \eqref{eq:prob:sub_cut}, which provides a near-optimal solution for max cut over given edge weights.}
\begin{equation}\label{eq:get_y_by_sgn}
\begin{aligned}
\mathbf{y}^{\beta,c}  = [ y^{\beta,c}_{1},\dots,y^{\beta,c}_{|\mathcal{V}^{\beta,c}|} ]^{\mathrm{T}} = \sgn\big( \mathbf{U}^{\beta,c} (\mathbf{\Sigma}^{\beta,c})^{\frac{1}{2}} \delta^{\beta,c} \big)\ .
\end{aligned}
\end{equation}
We then configure $\mathcal{V}^{\beta+1,2c-1}$ and $\mathcal{V}^{\beta+1,2c}$ as
\begin{equation}\label{eq:get_v_beta+1}
\begin{aligned}
\mathcal{V}^{\beta+1,2c-1} &= \{k^{\beta,c}_i|y^{\beta,c}_{i} = -1, \ i = 1,\dots,|\mathcal{V}^{\beta,c}|\}\ ,  \\
\mathcal{V}^{\beta+1,2c} &= \{k^{\beta,c}_i|y^{\beta,c}_{i} = +1, \ i = 1,\dots,|\mathcal{V}^{\beta,c}|\}\ .
\end{aligned}
\end{equation}
Next, we repeat the graph cut process on $\mathcal{V}^{\beta+1,2c-1}$ and $\mathcal{V}^{\beta+1,2c}$. Finally, when the process at the $\log_2(Z)$-th level is done, we have $Z$ disjoint groups, $\mathcal{V}^{\beta,1}\dots,\mathcal{V}^{\beta,Z}$, where $\beta=\log_2(Z)$, and the grouping decisions $\mathbf{z}$ is obtained as
\begin{equation}\label{eq:get_z}
\begin{aligned}
z_k = c , \ \forall k \in \mathcal{V}^{\beta,c},\ \beta = \log_2(Z),\ \forall c=1,\dots,Z. 
\end{aligned}
\end{equation}

Algorithm \ref{alg:recur_cut} summarizes the recursive graph cut procedure, namely $\mathrm{DoGraphCut}(\cdot)$. Its initial call at the root of the tree diagram is $\mathrm{DoGraphCut}(Z,\mathbf{W},\mathcal{V}'=\mathcal{V},\beta=0,c=1)$, as shown in Fig. \ref{fig:recur_cut}. Here, $Z$, $\mathbf{W}$ and $\mathcal{V}$ are the number of groups required, the adjacency matrix and the vertices of $\mathcal{G}$, respectively, as defined before. We simply denote the user grouping decisions computed based on this procedure as
\begin{equation}\label{eq:z_by_DoGraphCut}
\begin{aligned}
\mathbf{z} = \mathrm{DoGraphCut}(Z,\mathbf{W},\mathcal{V},0,1) \ .
\end{aligned}
\end{equation}

\begin{algorithm}[!t]
\caption{Recursive Graph Cut Procedure}\label{alg:recur_cut}
\begin{algorithmic}[1]
\PROCEDURE{$\mathrm{DoGraphCut}$}{$Z,\mathbf{W},\mathcal{V}',\beta,c$}
\IF{$\beta=\log_2(Z)$} 
\STATE Set $z_k = c $, $\forall k \in \mathcal{V}'$.
\ELSE
\STATE Set $\mathbf{W}^{\beta,c}$ as \eqref{eq:w_setup}, where $\mathcal{V}^{\beta,c}=\mathcal{V}'$.
\STATE Construct the problem in \eqref{eq:prob:sub_cut_sdp} and find $\hat{\mathbf{X}}^{\beta,c}$.
\STATE Compute $\mathcal{V}^{\beta+1,2c-1}$ and $\mathcal{V}^{\beta+1,2c}$ as \eqref{eq:svd}\eqref{eq:get_y_by_sgn}\eqref{eq:get_v_beta+1}.
\STATE $\mathrm{DoGraphCut}(Z,\mathbf{W},\mathcal{V}^{\beta+1,2c-1},\beta+1,2c-1)$.
\STATE $\mathrm{DoGraphCut}(Z,\mathbf{W},\mathcal{V}^{\beta+1,2c},\beta+1,2c)$.
\ENDIF
\RETURN
\ENDPROCEDURE
\end{algorithmic}
\end{algorithm}

\subsection{Design of the Graph-Evaluating Critic}\label{subsec:critic_nn_structure}
\begin{figure}[ht]
\begin{subfigure}[b]{1\columnwidth}
\centering
\includegraphics[scale=0.65]{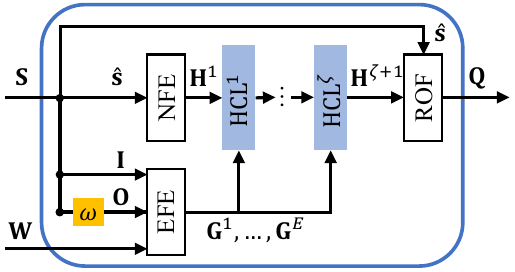}
\caption{The overall structure of the critic.}
\label{subfig:gnn_critic_nn_arch}
\end{subfigure}\\
\begin{subfigure}[b]{1\columnwidth}
\centering
\includegraphics[scale=0.65]{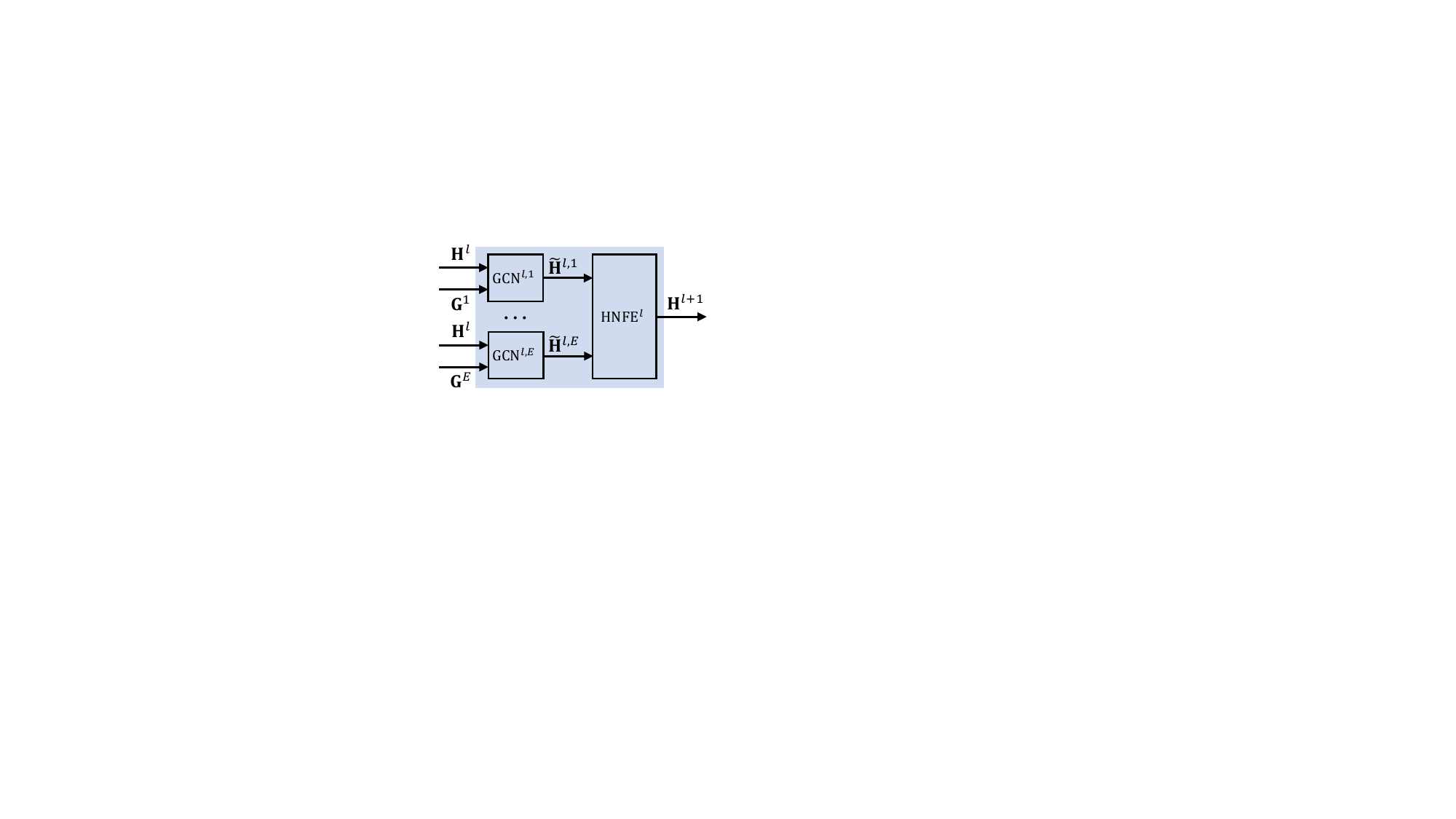}
\caption{The internal structure of HCL in each layer of the critic.}
\label{subfig:gnn_critic_hidden_layer}
\end{subfigure}
\caption{The structure of the critic.}
\vspace{-0.2cm}
\end{figure}
The critic $\mathbf{Q}(\cdot|\theta^Q)$ is a NN that approximates the expected value of the user throughput $\mathbf{r}$ for the given network states $\mathbf{S}$ and the given edge weights $\mathbf{W}$ in the graph as
\begin{equation}
\begin{aligned}
&\expt[\mathbf{r}|\mathbf{S},\mathbf{z}] \\
=       &\expt\big[[ r_{1},\dots,r_{K} ]^{\mathrm{T}}|\mathbf{S},\mathrm{DoGraphCut}(Z,\mathbf{W},\mathcal{V},0,1)\big] \\  
\approx & \mathbf{Q}(\mathbf{S},\mathbf{W}|\theta^Q) 
=[Q_1(\mathbf{S},\mathbf{W}|\theta^Q),\dots,Q_K(\mathbf{S},\mathbf{W}|\theta^Q)]^{\mathrm{T}} \ ,
\end{aligned}
\end{equation}
where $\mathbf{z}$ is the grouping decisions generated based on \eqref{eq:z_by_DoGraphCut} and $Q_k(\mathbf{S},\mathbf{W}|\theta^Q)$ is the $k$-th element in the output of $\mathbf{Q}(\mathbf{S},\mathbf{W}|\theta^Q)$ that approximates $r_k$, $\forall k$.
The overall structure of the critic is shown in Fig. \ref{subfig:gnn_critic_nn_arch}. 
To help the critic abstract information on the contention and interference from the network states, we apply the same pre-processing on the network states in the actor's design from Section \ref{subsec:actor_nn_structure}. 
Note that we define the vector and matrices to collect all users' pre-processed states and to simplify the notations as
\begin{equation}\label{eq:defi:user_states_reformulation_matrix_vector}
\begin{aligned}
\hat{\mathbf{s}} &\triangleq [s_{1,\hat{a}(1)},\dots,s_{K,\hat{a}(K)}] \ , \\
\mathbf{I} &\triangleq \big[I_{i,j}\big| I_{i,j}= s_{i,\hat{a}(j)} , \forall i\neq j; I_{k,k}=0,\forall k \big] \ , \\
\mathbf{O} &\triangleq \big[O_{i,j}\big| O_{i,j}= \omega(\mathbf{s}_i, \mathbf{s}_j|\theta^\omega) , \forall i\neq j; O_{k,k}=0,\forall k \big] \ .
\end{aligned}
\end{equation}
Here, $\hat{\mathbf{s}}$ contains all path losses from users to their associated APs, determining the receiving signal power and the packet duration of users.
$\mathbf{I}$ contains path losses from each user to all other users' associated APs, indicating the interference power.
$\mathbf{O}$ indicates how likely each pair of users can sense each other. Based on the above state pre-processing, the critic has a structure as
\begin{equation}\label{eq:critic_nn_structure}
\begin{aligned}
\mathbf{Q}(\mathbf{S},\mathbf{W}|\theta^Q) = \dot{\mathbf{Q}}(\hat{\mathbf{s}},\mathbf{I},\mathbf{O},\mathbf{W}|\theta^{\dot{Q}})|_{O_{i,j} = \omega(\mathbf{s}_i, \mathbf{s}_j|\theta^\omega),\forall i\neq j} \ ,
\end{aligned}
\end{equation}
where $\hat{\mathbf{s}}$ and $\mathbf{I}$ can be directly taken from $\mathbf{S}$, while $\mathbf{O}$ is computed based on $\mathbf{S}$ by using the inference NN $\omega(\cdot|\theta^\omega)$ defined in \eqref{eq:infer_hidden_target}. 
We note that $\omega(\cdot|\theta^\omega)$ is a shared part of both the actor and the critic.

The structure of $\dot{\mathbf{Q}}(\cdot)$ is designed based on GNNs that are flexible to the dimensions of its inputs, $\hat{\mathbf{s}}$, $\mathbf{I}$, $\mathbf{O}$ and $\mathbf{W}$, as follows.
First, we use a FNN, referred to as the edge feature embedder (EFE), to embed the user-pairwise state information, $\mathbf{I}$ and $\mathbf{O}$, and the actor-generated edge weights $\mathbf{W}$ into $E$ embedded edge features, e.g., $\mathbf{G}^{1},\dots,\mathbf{G}^{E}$, whose the $(i,j)$-th elements are
\begin{equation}\label{eq:efe}
\begin{aligned}
\relax [G^{1}_{i,j},\dots,G^{E}_{i,j}]^\mathrm{T} = \mathrm{EFE}(I_{i,j},O_{i,j},W_{i.j}) , \forall i\neq j \ , 
\end{aligned}
\end{equation}
and we set $G^{e}_{k,k}=0$, $\forall e,k$.
Also, we use a FNN, namely the node feature embedder (NFE), to embed the per-user-wise state information $\hat{\mathbf{s}}$ into a higher dimension as 
\begin{equation}\label{eq:nfe}
\begin{aligned}
\mathbf{H}^{1}\triangleq [\mathbf{h}^{1}_1,\dots,\mathbf{h}^{1}_K]^\mathrm{T} = [\mathrm{NFE}(s_{1,\hat{a}(1)}),\dots,\mathrm{NFE}(s_{K,\hat{a}(K)})]^\mathrm{T} ,
\end{aligned}
\end{equation}
where $\mathbf{h}^{1}_1,\dots,\mathbf{h}^{1}_K$ are $M$-dimensional vectors and $\mathbf{H}^{1}$ is a $K\times M$ matrix. 
Next, the hidden critic layers (HCLs) process the embedded features. There are $\zeta$ HCLs in the critic, e.g., $\mathrm{HCL}^{l}$ is the $l$-th HCL, $l=1,\dots,\zeta$. $\mathrm{HCL}^{l}$ has two inputs as 1) the embedded node features from the previous HCL's output $\mathbf{H}^{l}$ (or from the NFE's output when $l=1$), and 2) the embedded edge features from the EFE, $\mathbf{G}^1,\dots,\mathbf{G}^{E}$, i.e.,
\begin{equation}\label{eq:hcl}
\begin{aligned}
\mathbf{H}^{l+1} \triangleq [\mathbf{h}^{l+1}_1,\dots,\mathbf{h}^{l+1}_K]^\mathrm{T} = \mathrm{HCL}^{l}(\mathbf{H}^{l},\mathbf{G}^1,\dots,\mathbf{G}^{E}),
\end{aligned}
\end{equation}
where each HCL's internal structure is shown in Fig.~\ref{subfig:gnn_critic_hidden_layer}.
Specifically, we use $E$ graph convolutional networks (GCNs)~\cite{kipf2017semi}\footnote{we use GCNs to construct the critic because most inputs of the critic in \eqref{eq:critic_nn_structure} are user-pairwise features (e.g., $\mathbf{I}$, $\mathbf{O}$ and $\mathbf{W}$) that can be viewed as edge features on a graph, as GCNs dedicated for. Alternatively, other GNNs can take edge features as inputs can be used, e.g., GNNs listed in \cite{pytorchgeo_gnnlist}. However, what the optimal GNN structure is for the critic is not the focus of this work.}, e.g., $\mathrm{GCN}^{l,e}$, $e=1,\dots,E$, in $\mathrm{HCL}^{l}$ to aggregate $\mathbf{H}^{l}$ and $\mathbf{G}^{e}$ into hidden node features as
\begin{equation}\label{eq:gcn}
\begin{aligned}
 \Tilde{\mathbf{H}}^{l,e}\triangleq &[\Tilde{\mathbf{h}}^{l,e}_1,\dots,\Tilde{\mathbf{h}}^{l,e}_K]^\mathrm{T} =  \mathrm{ReLU}\big(\mathrm{GCN}^{l,e}(\mathbf{H}^{l},\mathbf{G}^{e})\big) \\
= &\mathrm{ReLU}\big((\mathbf{D}^{l,e})^{-\frac{1}{2}}(\mathbf{G}^{e}+\mathbb{I}_K)(\mathbf{D}^{l,e})^{-\frac{1}{2}}\mathbf{H}^{l} \Theta^{l,e}\big) \ ,
\end{aligned}
\end{equation}
where $\mathrm{ReLU}(\cdot)$ is the rectified linear activation function and $\Theta^{l,e}$ are trainable parameters of $\mathrm{GCN}^{l,e}$.
Here, $\mathbb{I}_K$ is a $K\times K$ identity matrix and $\mathbf{D}^{l,e}$ is the diagonal degree matrix of $\mathbf{G}^{e}+\mathbb{I}_K$ (i.e., $D^{l,e}_{i,i}= \sum_{j=1}^{K} G^e_{i,j}+1$) in \eqref{eq:gcn}. Next, a FNN, namely hidden node feature embedder (HNFE), aggregates each user's hidden node features from all GCNs in $\mathrm{HCL}^{l}$ as
\begin{equation}\label{eq:hfne}
\begin{aligned}
&\mathbf{H}^{l+1} 
\triangleq [\mathbf{h}^{l+1}_1,\dots,\mathbf{h}^{l+1}_K]^\mathrm{T}\\
=& \Big[\mathrm{HNFE}^{l}
\Big(
\big[(\Tilde{\mathbf{h}}^{l,1}_1)^\mathrm{T},\dots,(\Tilde{\mathbf{h}}^{l,E}_1)^\mathrm{T}\big]^\mathrm{T}
\Big),\\
&\dots,\mathrm{HNFE}^{l}
\Big(
\big[(\Tilde{\mathbf{h}}^{l,1}_K)^\mathrm{T},\dots,(\Tilde{\mathbf{h}}^{l,E}_K)^\mathrm{T}\big]^\mathrm{T}
\Big)\Big]^\mathrm{T} \ , \forall l \ , \\
\end{aligned}
\end{equation}
where $\mathrm{HNFE}^l$ is the HNFE in the $l$-th layer and $\mathbf{H}^{l+1}$ is forwarded to the next HCL.
Last, a FNN is used as a readout function (ROF) to map the node features from the output of the last HCL and the initial per-user-wise network states $\hat{\mathbf{s}}$ into the approximated user throughput as
\begin{equation}\label{eq:rof}
\begin{aligned}
&\mathbf{Q}(\mathbf{S},\mathbf{W}|\theta^Q) = \dot{\mathbf{Q}}(\hat{\mathbf{s}},\mathbf{I},\mathbf{O},\mathbf{W}|\theta^{\dot{Q}})\\
= &\big[\mathrm{ROF}(\mathbf{h}^{\zeta+1}_{1},s_{1,\hat{a}(1)}),\dots,\mathrm{ROF}(\mathbf{h}^{\zeta+1}_{K},s_{K,\hat{a}(K)})\big]^\mathrm{T} \ ,
\end{aligned}
\end{equation}
where $\mathbf{h}^{\zeta+1}_{k}$, $\forall k $, are computed as \eqref{eq:efe}-\eqref{eq:hfne}.

\subsection{The Flow of the AC-GRL Algorithm}\label{subsec:grl_algorithm_flow}
Finally, we explain the flow of the AC-GRL algorithm that trains NNs in the above components, where the initial values of all NN parameters are randomly initialized.

\begin{algorithm}[!t]
\caption{Actor-Critic Graph Representation Learning}\label{alg:ac_grl}
\begin{algorithmic}[1]
\STATE Randomly initialize NNs' parameters, $\theta^\omega$, $\theta^{\dot{Q}}$ and $\theta^{\dot{\mu}}$.
\FOR{step $n$ = $1,\dots,N$}\label{alg:line:training_inference_nn_start}
    \STATE Generate a network randomly, and measure $\mathbf{S}$ and $\hat{\mathbf{O}}$.
    \STATE Compute $\nabla_{\theta^{\omega}} L(\theta^{\omega})$ as \eqref{eq:grad_loss_omega} and perform the SGD on the inference NN as\\ \qquad\qquad $\theta^{\omega}\leftarrow \theta^{\omega} - \eta \nabla_{\theta^{\omega}} L(\theta^{\omega})$.
\ENDFOR\label{alg:line:training_inference_nn_end}
\FOR{step $n$ = $1,\dots,N$}\label{alg:line:training_ac_nn_start}
    \STATE Generate a network randomly, and measure $\mathbf{S}$.
    \STATE Generate $\mathbf{W}$ using the graph-constructing actor \eqref{eq:actor_nn_structure}.
    \STATE With probability $\nu$, randomly set the values in $\mathbf{W}$.
    \STATE Compute $\mathbf{z}$ as \eqref{eq:z_by_DoGraphCut} using the graph cut procedure.
    \STATE Execute $\mathbf{z}$ in the network and measure $\mathbf{r}$.
    \STATE Compute $\nabla_{\theta^{\dot{Q}}}  L(\theta^{Q})$ as \eqref{eq:grad_loss_c} using the graph-evaluating critic and perform the SGD on the critic as \\
    \qquad\qquad $\theta^{\dot{Q}}\leftarrow \theta^{\dot{Q}} - \eta \nabla_{\theta^{\dot{Q}}}  L(\theta^{Q})$.
    \STATE Compute $\nabla_{\theta^{\dot{\mu}}}  L(\theta^{\mu})$ as \eqref{eq:grad_loss_a} using the actor and the critic and perform the SGD on the actor as  \\
    \qquad\qquad $\theta^{\dot{\mu}}\leftarrow \theta^{\dot{\mu}} - \eta \nabla_{\theta^{\dot{\mu}}}  L(\theta^{\mu})$.
\ENDFOR\label{alg:line:training_ac_nn_end}
\STATE \textbf{return} $\theta^\omega$, $\theta^{\dot{Q}}$ and $\theta^{\dot{\mu}}$ for online fine-tuning.
\end{algorithmic}
\end{algorithm}

\subsubsection{Pre-training of the Inference NN}
We first train the inference NN $\omega(\cdot|\theta^\omega)$ that is the shared part of the actor and the critic. In each step, we simulate a random realization of the wireless network where we measure the network states $\mathbf{S}$. We also acquire whether or not each pair of users can sense each other as a $K \times K$ matrix of binary indicators, defined as
\begin{equation}\label{eq:defi:true_value_hidden_users}
\begin{aligned}
\hat{\mathbf{O}}\triangleq[\hat{O}_{i,j} | \hat{O}_{i,j} = \mathbf{1}_{\{\text{user $j$ can sense user $i$}\}}, \forall i\neq j ; \hat{O}_{k,k}=0,\forall k] .
\end{aligned}
\end{equation}
Note that $\hat{\mathbf{O}}$ is only measured in offline simulations and does not need to be measured when deploying our methods in a real-world network.
Then, $\omega(\cdot|\theta^\omega)$'s parameters are optimized by minimizing a loss function as the cross entropy between $\hat{O}_{i,j}$ and its inferred expected value $O_{i,j}$ as
\begin{equation}\label{eq:loss_inference_nn}
\begin{aligned}
&L(\theta^\omega) = \mathop{\mathbb{E}} \Big[\sum_{i\neq j}-\hat{O}_{i,j}\log(\omega(\mathbf{s}_i,\mathbf{s}_j|\theta^\omega))\\
&\qquad\qquad\qquad-(1-\hat{O}_{i,j})\log(1-\omega(\mathbf{s}_i, \mathbf{s}_j|\theta^\omega)) \Big]\ ,
\end{aligned}
\end{equation}
where we update $\theta^\omega$ using the stochastic gradient descent (SGD) method \cite{kingma2014adam} based on the gradient of \eqref{eq:loss_inference_nn} as
\begin{equation}\label{eq:grad_loss_omega}
\begin{aligned}
\nabla_{\theta^{\omega}}  L(\theta^{\omega})=\sum_{i\neq j} \big(\frac{-\hat{O}_{i,j}}{O_{i,j}} - \frac{1-\hat{O}_{i,j}}{1-O_{i,j}}(-1)\big)\nabla_{\theta^{\omega}}\omega(\mathbf{s}_i, \mathbf{s}_j|\theta^\omega) .
\end{aligned}
\end{equation}
The training steps of the inference NN are repeated $N$ times.

\subsubsection{Main Training Process of the AC-GRL Algorithm}
After the inference NN is trained, we train the remaining parameters of the actor and the critic. The training process of the actor and the critic has been shown in Fig. \ref{fig:actor_critic_algorithm}.
We measure network states $\mathbf{S}$ from a randomized realization of the wireless network at the start of each training step. Then, the graph-constructing actor generates edge weights $\mathbf{W}$ as \eqref{eq:actor_nn_structure}. 
With a probability of $\nu$ ($\nu\in [0,1]$), we set the off-diagonal elements of $\mathbf{W}$ with random numbers in $[0,1]$ to explore good edge weights. 
After that, the graph cut procedure computes the grouping decisions $\mathbf{z}$ based on $\mathbf{W}$ by calling $\mathrm{DoGraphCut}(\cdot)$ in Algorithm \ref{alg:recur_cut} as \eqref{eq:z_by_DoGraphCut}. We then measure the user throughput $\mathbf{r}$ from the network for the given $\mathbf{z}$. 
The measured user throughput is used to optimize the graph-evaluating critic that approximates the throughput for given network states and edge weights. Thus, the loss function of the critic is the difference between the estimated user throughput and the actual measured one as
\begin{equation}\label{eq:loss_c_original_state}
\begin{aligned}
L(\theta^Q) =
\mathop{\mathbb{E}} \Big[ \sum_{k=1}^{K} \big(r_k  -Q_k(\mathbf{S},\mathbf{W}|\theta^Q)  \big)^2 \Big] \ ,
\end{aligned}
\end{equation}
whose gradient of $L(\theta^{Q})$ with respect to $\theta^{\dot{Q}}$ is
\begin{equation}\label{eq:grad_loss_c}
\begin{aligned}
&\nabla_{\theta^{\dot{Q}}} L(\theta^{Q}) = \sum_{k=1}^{K}\Big\{2 \big[r_k- \dot{Q}_k(\hat{\mathbf{s}},\mathbf{I},\mathbf{O},\mathbf{W}|\theta^{\dot{Q}})\big]\\
&\qquad\qquad \times\nabla_{\theta^{\dot{Q}}}\dot{Q}_k(\hat{\mathbf{s}},\mathbf{I},\mathbf{O},\mathbf{W} |\theta^{\dot{Q}})\Big\} \Big|_{O_{i,j} = \omega(\mathbf{s}_i, \mathbf{s}_j|\theta^\omega),\forall i\neq j}\ . 
\end{aligned}
\end{equation}
Also, the critic is used to optimize the actor. Specifically, the actor is optimized to generate the best edge weights that maximize the worst-case user throughput in \eqref{eq:prob:user_grouping:graph_cut:adaptive_edge_weighting}, whose value is expressed as a loss function based on the critic as
\begin{equation}\label{eq:loss_a_original_state}
\begin{aligned}
L(\theta^{\mu}) = \mathop{\mathbb{E}} \Big[- \min_k Q_k\big(\mathbf{S},\mathbf{W}|\theta^Q)|_{W_{i,j} = \mu(\mathbf{s}_i,\mathbf{s}_j|\theta^\mu)} \Big] \ ,
\end{aligned}
\end{equation}
where $\min_k Q_k(\cdot)$ is the worst-case user throughput approximated by the critic and we can derive the gradients of $L(\theta^{\mu})$ with respect to $\theta^{\dot{\mu}}$ as
\begin{equation}\label{eq:grad_loss_a}
\begin{aligned}
&\nabla_{\theta^{\dot{\mu}}}  L(\theta^{\mu})
=  - \nabla_{\mathbf{W}}\min_k Q_k\big(\mathbf{S},\mathbf{W}|\theta^Q) \times \nabla_{\theta^{\dot{\mu}}} \mathbf{W} = \\
& - \sum_{i \neq j}\frac{\partial \min_k\dot{Q}_k(\hat{\mathbf{s}},\mathbf{I},\mathbf{O},\mathbf{W}|\theta^{\dot{Q}})}{\partial W_{i,j}} \\
&\times \nabla_{\theta^{\dot{\mu}}}\dot{\mu}(s_{j,\hat{a}(j)},s_{i,\hat{a}(j)},s_{i,\hat{a}(i)},O_{i,j}|\theta^{\dot{\mu}})\big|_{O_{i,j} = \omega(\mathbf{s}_i, \mathbf{s}_j|\theta^\omega),\forall i\neq j}. \\
\end{aligned}
\end{equation}
We use the gradients in \eqref{eq:grad_loss_c} and \eqref{eq:grad_loss_a} to update the actor and the critic's parameters based on SGD, respectively. Here, we repeat the training steps for the actor and the critic for $N$ times.
Algorithm \ref{alg:ac_grl} summarizes the NNs' training process, where lines \ref{alg:line:training_inference_nn_start}-\ref{alg:line:training_inference_nn_end} are the inference NN's training process and lines \ref{alg:line:training_ac_nn_start}-\ref{alg:line:training_ac_nn_end} are the actor and the critic's training process.
Note that $\eta$ is the learning rate of the NNs.

We study the complexity of the AC-GRL algorithm in Algorithm \ref{alg:ac_grl} with regard to the number of APs $A$, the number of users $K$, and the number of group $Z$.
The detailed complexity analysis of each component in the algorithm is discussed in the appendix.
In summary, the pre-training process in the AC-GRL algorithm has a complexity of $\mathcal{O}(K^2A^2)$ in each step, while the main training process has a complexity of $\mathcal{O}(K^2A^2 +(Z-1) K^{3.5} + K^2)$ in each step (both processes are polynomial to $K$ and $A$ and the latter is also linear to $Z$).

\section{Online Fine-Tuning Architecture of AC-GRL}\label{sec:online_arch}
\begin{figure}[!t]
\centering
\includegraphics[scale=0.75]{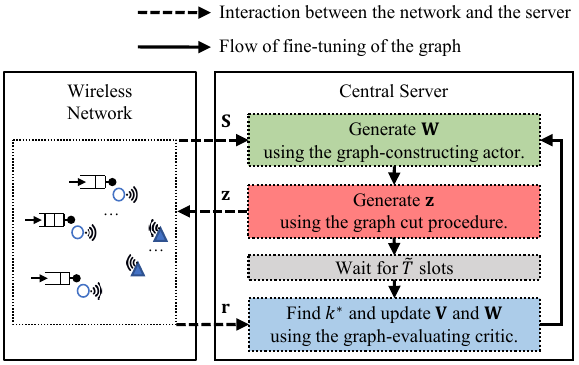}
\caption{The online fine-tuning architecture of AC-GRL.}
\label{fig:wifi_online_architecture}
\vspace{-0.2cm}
\end{figure}
In this section, we explain the online architecture, as shown in Fig. \ref{fig:wifi_online_architecture}, that continuously fine-tunes the graph's edge weights for a given static realization of the wireless network as well as for dynamic networks, e.g., due to user mobility.
Note that we assume that Algorithm \ref{alg:ac_grl} has returned the trained parameters of the inference NN, the actor and the critic in offline training. We will not make any updates on them in the online architecture.

Specifically, to fine-tune the edge weights for a given network, we measure the network states $\mathbf{S}$ (user-to-AP path losses) at the initialization of the network and infer contending user probability $O_{i,j}$, $\forall i \neq j$, as \eqref{eq:infer_hidden_target} before the transmissions start. 
Next, at the first RAW slot, $t=1$, we compute edge weights $\mathbf{W}$ using the actor as \eqref{eq:actor_nn_structure} and then use the computed $\mathbf{W}$ to generate the user grouping decisions $\mathbf{z}$ as \eqref{eq:z_by_DoGraphCut}.
The decisions will remain unchanged for $\Tilde{T}$ slots (we assume $\Tilde{T}>Z$). After that, at slot $t=1+\Tilde{T}$, we measure each user's throughput $\Tilde{r}_k$, $k=1,\dots,K$, over the past $\Tilde{T}$ slots and find the worst-case user's index as
\begin{equation}\label{eq:worst-case_online}
\begin{aligned}
\Tilde{r}_k\leftarrow \frac{1}{\Tilde{T}} \sum_{t'=t-\Tilde{T}}^{t-1} u_k(t') \ , \  k^* \leftarrow \arg\min_k \Tilde{r}_k \ .
\end{aligned}
\end{equation}The elements in $\mathbf{W}$ are then fine-tuned based on the worst-case user's index and the critic $\mathbf{Q}$ in \eqref{eq:critic_nn_structure} by minimizing the following loss function,
\begin{equation}\label{eq:loss_weight_online_arch}
\begin{aligned}
L(\mathbf{W}) =  - Q_{k^*}(\mathbf{S},\mathbf{W}|\theta^Q) \ .
\end{aligned}
\end{equation}
In order to prevent the updated weight values from going out of the bounds on weights, i.e., $[0,1]$, we update an intermediate $K\times K$ matrix $\mathbf{V}$ instead, where
\begin{equation}\label{eq:defi:v}
\begin{aligned}
\mathbf{V} &\triangleq \big[V_{i,j}\big| V_{i,j}= \mathrm{Sigmoid}^{-1}(W_{i,j}) , \forall i\neq j; V_{k,k}=0,\forall k \big] .
\end{aligned}
\end{equation}
Here, $\mathrm{Sigmoid}^{-1}$ is the inverse of the $\mathrm{Sigmoid}$ function with the output range in $[0,1]$.
Then, $\mathbf{V}$ is updated based on the gradient of the critic as
\begin{equation}\label{eq:gradient_loss_w_online}
\begin{aligned}
\nabla_{\mathbf{V}}  L(\mathbf{W}) 
&=  - \nabla_{\mathbf{V}}Q_{k^*}(\mathbf{S},\mathbf{W}|\theta^Q) \\
&=  - \nabla_{\mathbf{W}}Q_{k^*}(\mathbf{S},\mathbf{W}|\theta^Q) \times \nabla_{\mathbf{V}}\mathbf{W} \ ,
\end{aligned}
\end{equation}
where $\nabla_{\mathbf{V}}\mathbf{W}$ can be computed based on the derivative of the $\mathrm{Sigmoid}$ function in \eqref{eq:defi:v}.
Then, $\mathbf{V}$ and  $\mathbf{W}$ are updated as
\begin{equation}\label{eq:update_w_online}
\begin{aligned}
\mathbf{V} \leftarrow \mathbf{V} - \eta \nabla_{\mathbf{V}}  L(\mathbf{W}) \ , 
\mathbf{W} \leftarrow \mathrm{Sigmoid} (\mathbf{V} ) \ ,
\end{aligned}
\end{equation}
where $\eta$ is the learning rate. In scenarios where users move, network states $\mathbf{S}$ can vary over the past $\Tilde{T}$ slots as the path losses between users and APs vary due to changes in their distances.
We constantly update $\mathbf{S}$ to the latest values measured from the most recent packets sent by users. The updated states $\mathbf{S}$ are used to partly regenerate the edge weights $\mathbf{W}$ as 
\begin{equation}\label{eq:update_w_online_refresh}
\begin{aligned}
W_{i,j} \leftarrow (1-\lambda) W_{i,j} +  \lambda \cdot \mu(\mathbf{s}_i,\mathbf{s}_j|\theta^\mu) \ , \ \forall i , j \ ,
\end{aligned}
\end{equation}
where $\lambda$ is the weight regeneration rate.
Based on the updated $\mathbf{W}$, we recompute the grouping decisions as \eqref{eq:z_by_DoGraphCut} and measure user throughput again for the next $\Tilde{T}$ slots. At slot $t+\Tilde{T}$, the worst-case user during time from $t$ to $t+\Tilde{T}$ is found as \eqref{eq:worst-case_online} and we update $\mathbf{V}$ and $\mathbf{W}$ as \eqref{eq:loss_weight_online_arch}-\eqref{eq:update_w_online_refresh}. The above process repeats every $\Tilde{T}$ slots until the network stops.

\section{Evaluation of Proposed Methods}\label{sec:evaluation}
This section provides the simulation results that evaluate our proposed methods.
\subsection{Simulation Configurations}\label{subsec:ns3_simulation_config}
We use NS-3 \cite{ns3,tian2016implementation} to simulate Wi-Fi HaLow networks. 
All devices are located in a 2km~$\times$~2km squared area centered at the coordinates $(0,0)$ meter.
We assume there are $4$ APs ($A=4$), and they are located at the grid in the simulated area \cite{aust2015outdoor,adhiatma2020ieee}, i.e., they are at the coordinates $(500,500)$, $(-500,500)$, $(500,-500)$ and $(-500,-500)$ meters.
Unless specifically stated, we assume that users are static, and we set the number of users $K$ to $20$.
Each user is randomly distributed in the simulated area, where each coordinate is randomly selected from the uniform distribution with an interval $[-1000,1000]$ meters. The duration of a RAW slot, $\Delta_\text{0}$, is configured as $10$~milliseconds. The channel bandwidth $B$ is set as $1$~MHz at a $1$~GHz carrier frequency. 
The transmission power of users is set as $\mathbf{P}_\mathrm{0}=0$~dBm, and the noise power is set as $\mathbb{N}_\mathrm{0}B =-94$~dBm.
The path losses between any two devices (including APs and users) follow Friis model \cite{friis1946note} as $10\log10((\frac{\mathrm{c}}{1\text{ GHz}})^2/(4\pi d)^2)$ in dB, where $\mathrm{c}$ is the speed of light and $d$ is the distance in meters between two devices. The threshold on path losses where a device can sense a transmission, $\Tilde{s}_\mathrm{max}$, is set to $95$~dB. The packet size $L$ is 800 bits, the maximum queue size is $5$, and each user has a stationary Poisson packet arrival process with intervals of 20 milliseconds. The maximum decoding error probability $\epsilon_{\max}$ for each user is $10^{-5}$ when no interference exists. The decoding error probability of each transmission is computed using the same equation in \eqref{eq:tx_error_for_mcs_wi-fi}, where $\phi$ is the signal-to-interference-plus-noise of each transmission instead. 

All FNNs in the actor and the critic have one input layer, one output layer, and two hidden layers. The activation functions of all hidden layers in FNNs are set as $\mathrm{ReLU}$ functions. The size of each layer and the output activation functions (OAFs) of FNNs are listed in Table \ref{tab:configurations_of_FNNs}, where $M=E=5$. The number of HCLs in the critic is $\zeta = 3$. Trainable parameters of GCNs, $\Theta^{l,e}$, has the size of $M\times M$, $\forall l,e$.
The exploration rate $\nu$ is $10\%$, the NNs' learning rate $\eta$ is $10^{-4}$ and $10^{-3}$ in Algorithm \ref{alg:ac_grl} and the online architecture, respectively, and the weight regeneration rate $\lambda$ in the online architecture is $0.1$. The number of steps $N$ in Algorithm \ref{alg:ac_grl} is $1000$.
\begin{table}[t]
\vspace{0.2cm}\small
\caption{Configurations of FNNs in the Actor and the Critic}
\label{tab:configurations_of_FNNs}
\begin{minipage}{\columnwidth}
\begin{center}
\begin{tabular}{|c|c|c|}
\hline
NNs             & Dimensions of layers & OAFs   \\ \hline
$\omega$        &   $2A,20A,20A,1$                   &      $\mathrm{Sigmoid}(\cdot)$                        \\ \hline
$\dot{\mu}$     &   $4,40,40,1$              &               $\mathrm{Sigmoid}(\cdot)$               \\ \hline
 $\mathrm{EFE}$             &      $3,30,30,E$                &              $\mathrm{ReLU}(\cdot)$                \\ \hline
 $\mathrm{NFE}$             &         $1,10,10,M$             &           $\mathrm{ReLU}(\cdot)$                       \\ \hline
 $\mathrm{HNFE}^l$, $\forall l$           &        $EM,10EM,10EM,M$              &      $\mathrm{ReLU}(\cdot)$                            \\ \hline
 $\mathrm{ROF}$             &        $M+1,10(M+1),10(M+1),1$                 &     None                        \\ \hline
\end{tabular}
\vspace{-0.2cm}
\end{center}
\end{minipage}
\end{table}

\subsection{Compared Methods in Simulations}
We explain the implementation of the compared methods other than the proposed methods, including heuristics, Markov-model-based, graph-based and ML-based approaches.
\subsubsection{Applying Heuristics}
We compare our method to the heuristics method that simply randomly allocates users in each group with equal probability, referred to as the ``RAND'' scheme. We also compare a method that uniformly balances the number of users in each group, referred to as the ``UNIF'' scheme. Specifically, we sort the users based on their associated AP's index as $k'_1,\dots,k'_K$, where $\hat{a}(k'_1)\leq \hat{a}(k'_2) \leq\cdots\leq \hat{a}(k'_K)$. Then, we set $z_{k'_i} = (i\ \mathbf{mod}\ Z)+1$, $\forall k$, which balances the number of users in different groups for each AP.

\subsubsection{Applying Markov-model-based Approach}
We compare the Markov-model-based approach \cite{chang2018traffic,kai2019energy} to the user grouping problem, where the Markov models in \cite{bianchi2000performance} are used to estimate each user's throughput. Note that works in \cite{chang2018traffic,kai2019energy} assume that all users can sense each other, i.e., no hidden users.
The grouping decisions are generated by using the iterative algorithm developed in \cite{chang2018traffic}. Specifically, In each iteration, the algorithm selects the best user among all un-grouped users and the best grouping decision of the selected user to maximize the worst-case user throughput in grouped users. Then, the selected user is marked as grouped, and another un-grouped user will be selected in the next iteration until all users are grouped. We refer to the above method as the ``MC-based'' scheme in simulations.

\subsubsection{Applying Graph-based Approach}
The graph-based approach to the grouping problem is also compared. Note that the works in \cite{chen2022energy} formulate the user grouping problem as a graph coloring problem in which edges connect the users if the same AP can detect them.
Then, connected users must have different colors/groups and be assigned different RAW slots.
This approach requires many groups in dense networks as users are densely connected in the above graph construction. This is unsuitable for the problem in \eqref{eq:prob:user_grouping} where the number of groups $Z$ is limited. Then, we consider using the graph-max-cut-based schemes in \cite{chang2009multicell,liang2018graph}. They compute the user grouping decisions as \eqref{eq:prob:user_grouping:max_cut} but based on manually designed rules to generate the graph's edge weights.
For example, since the user throughput is mostly affected by the contention and interference in the network, we can design the edge weights as $W_{i,j}=O_{i,j}$, $\forall i \neq j$, where $O_{i,j}$ is computed using the trained inference NN as \eqref{eq:infer_hidden_target}, indicating how likely user $j$ can sense user $i$, i.e., how likely user $i$ can trigger the CSMA/CA contention of user $j$. 
Alternatively, we can set the edge weights as $W_{i,j}=1-O_{i,j}$, $\forall i \neq j$, to indicate how likely users are hidden or how likely they will make concurrent transmissions causing interference.
Also, we can set edge weights $W_{i,j}$, $\forall i \neq j$, to indicate the interference caused by user $i$ to user $j$'s associated AP as $W_{i,j}=\phi'_{i,j} = (\mathbf{P}_\mathrm{0}/s_{j,\hat{a}(j)})/(\mathbb{N}_\mathrm{0}B+\mathbf{P}_\mathrm{0}/s_{i,\hat{a}(j)})$ (where $s_{i,j}$ are unnormalized). We refer to the above graph-max-cut-based schemes that set $W_{i,j}$ to $O_{i,j}$, $1-O_{i,j}$ and $\phi'_{i,j}$, $\forall i \neq j$, as the ``MCON'', ``MHID'' and ``MINT'' schemes, respectively.

\subsubsection{Applying ML-based Approach}
We also compare the ML-based approach when solving the problem, where a NN is trained to directly generate grouping decisions $\mathbf{z}$ based on network states $\mathbf{S}$. 
For example, we use the random-edge GNN (REGNN) design in \cite{eisen2020optimal}, whose structure of each layer is
\begin{equation}\label{eq:defi:regnn}
\begin{aligned}
 \mathbf{g}^{l+1}= \sigma(\sum_{n=1}^{\xi}\alpha_{n}^l (\mathbf{I})^n\mathbf{g}^l) \ , \ l=1,\dots,\chi-1.
\end{aligned}
\end{equation}
where $\mathbf{I}$ is defined in \eqref{eq:defi:user_states_reformulation_matrix_vector} and $\mathbf{g}^i$ are $K$-dimensional vectors ($\mathbf{g}^1$ is set as $\hat{\mathbf{s}}$ defined in \eqref{eq:defi:user_states_reformulation_matrix_vector}). 
Here, $\chi=5$ is the number of layers, $\xi=5$ is the graph filter coefficient, $\alpha_{n}^l$, $\forall n , l$, are trainable parameters in each layer, and $\sigma$ is the activation function (e.g., $\mathrm{ReLU}(\cdot)$). Note that REGNNs only output one scalar feature for each user. Thus, we encode user grouping decisions $z_k$, $k=1,\dots,K$, as a binary number, e.g., $z_k=1,2,3,4$ as $z_k=00,01,10,11$, and use two REGNNs with different parameters outputs two bits of the grouping decisions separately. 
We can also use message-passing GNN (MPGNN) studied in \cite{shen2020graph} to generate the user grouping decisions, e.g.,
\begin{equation}\label{eq:defi:mpgnn}
\begin{aligned}
 \mathbf{g}^{l+1}_i= \sigma\Big( \mathrm{FNN2}\big( \mathbf{g}^{l}_i, \max_{j\neq i}\mathrm{FNN1}(\mathbf{g}^{l}_i, s_{i,\hat{a}(j)})\big)\Big), \\
 \ l=1,\dots,\chi-1, i=1,\dots,K \ , 
\end{aligned}
\end{equation}
where $\mathrm{FNN1}$ and $\mathrm{FNN2}$ are two FNNs, $\chi=3$ is the number of layer, $\mathbf{g}^{l}_i$ is the user-wise feature of $i$-th user in the $l$-th layer ($\mathbf{g}^{1}_i=\mathbf{s}_i$). 
We use policy gradient \cite{eisen2020optimal,sutton1999policy} to train the REGNN and the MPGNN.

\subsection{Performance of Proposed AC-GRL Algorithm}
We evaluate the AC-GRL algorithm in Algorithm \ref{alg:ac_grl} below.
\subsubsection{Training of the Inference NN}
We show the training information of the inference NN in Fig. \ref{fig:plot_infer}, i.e., the first part of Algorithm \ref{alg:ac_grl} (lines \ref{alg:line:training_inference_nn_start}-\ref{alg:line:training_inference_nn_end}). In Fig. \ref{subfig:plot_infer_loss_accu}, the value of the loss function in \eqref{eq:loss_c_original_state} decreases over training steps and converges around $400$ steps. Also, we measure the accuracy of the inference NN during training. For example, we sample the binary value $O'_{i,j}$ from $O_{i,j}$ computed as \eqref{eq:infer_hidden_target} and measure the probability that $O'_{i,j}$ is equal to its true value $\hat{O}_{i,j}$, as shown in Fig. \ref{subfig:plot_infer_loss_accu}. The results show the accuracy of the inference NN also converges around $400$ steps. Further, we measure the accuracy of the NN separately for two possible values of $\hat{O}_{i,j}$ in Fig. \ref{subfig:plot_infer_accu_0_1}. The results show that the accuracy converges in either case, implying that the inference NN makes no biased guess on the value of $\hat{O}_{i,j}$, e.g., simply guessing one of the possible values. Thus, the inference NN is well-trained and can be used in the following part of the algorithm.

\begin{figure}[t]
\begin{subfigure}[b]{1\columnwidth}
\centering
\includegraphics[scale=0.8]{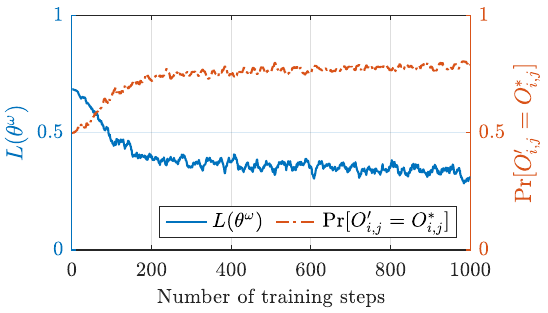}
\caption{The loss and the accuracy of the inference NN during training.}
\label{subfig:plot_infer_loss_accu}
\end{subfigure}\\
\begin{subfigure}[b]{1\columnwidth}
\centering
\includegraphics[scale=0.8]{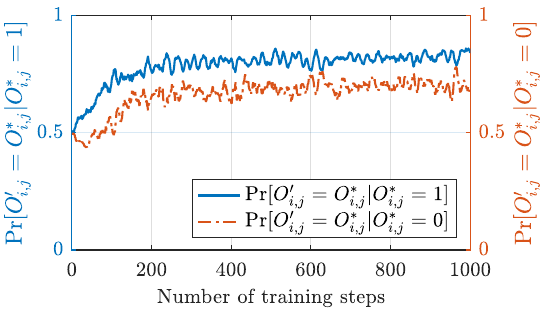}
\caption{The accuracy of the inference NN for two different cases of $\hat{O}_{i,j}$, $\forall i,j$, during training.}
\label{subfig:plot_infer_accu_0_1}
\end{subfigure}
\caption{The training information of the inference NN $\omega$.}
\label{fig:plot_infer}
\vspace{-0.2cm}
\end{figure}

\subsubsection{Training of the Actor and the Critic}
Next, we evaluate the training of the actor and the critic in Algorithm \ref{alg:ac_grl} (lines \ref{alg:line:training_ac_nn_start}-\ref{alg:line:training_ac_nn_end}). We measure the worst-case user throughput, $\min_k r_k$, and the average throughput of users, $\frac{1}{K} \sum_k r_k$, during training in Fig. \ref{subfig:plot_training_min_rate} and Fig. \ref{subfig:plot_training_sum_rate}, respectively. 
We compare our AC-GRL method (with legend ``Proposed'') with the methods directly using the REGNN or the MPGNN as the user grouping policy.
The results show that our method converges after $500$ training steps, achieving over $80\%$ more throughput than two other schemes without degradation in the average throughput of users. 
This is because the compared NNs, REGNN and MPGNN, aggregate all neighboring users' information of a given user, e.g., due to the matrix-vector multiplication $(\mathbf{I})^n\mathbf{g}^l$ and $\max$ function in \eqref{eq:defi:regnn} and \eqref{eq:defi:mpgnn}, respectively. As a result, the contention and interference information in each user pair is lost in aggregation in these methods. Therefore, they cannot express the user-pairwise correlation between grouping decisions, i.e., how likely a pair of users should be assigned in the same slot. 
However, our method uses a max cut process to generate the user grouping decisions, where the correlation between each pair of users' decisions is retained and indicated by the edge weights between each user pair. 

In the following simulations, we use the above-trained actor and critic (trained with $K=20$ users and $Z=4$ groups) to generate edge weights and grouping decisions without further updating the parameters of the actor and the critic.
\begin{figure}[t]
\begin{subfigure}[b]{1\columnwidth}
\centering
\includegraphics[scale=0.8]{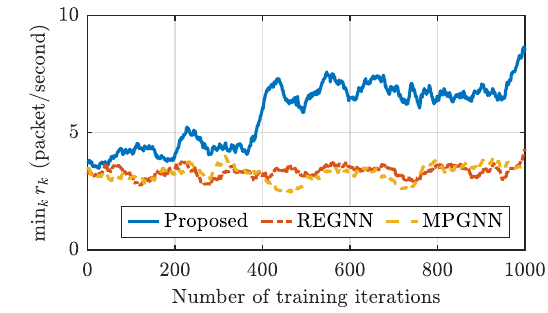}
\caption{The worst-cause user throughput during training.}
\label{subfig:plot_training_min_rate}
\end{subfigure}\\
\begin{subfigure}[b]{1\columnwidth}
\centering
\includegraphics[scale=0.8]{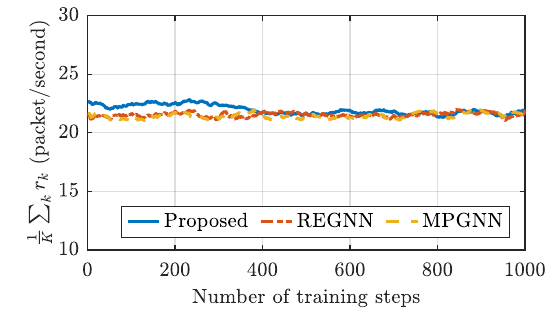}
\caption{The averaged user throughput during training.}
\label{subfig:plot_training_sum_rate}
\end{subfigure}
\caption{The training information of the actor and the critic.}
\label{fig:plot_training}
\vspace{-0.2cm}
\end{figure}

\subsection{Performance of the Trained Actor}
\subsubsection{Comparison to Other Methods}
Next, we use the trained actor to generate the graph's edge weights and compare the user grouping decisions made by the graph cut procedure with the decisions made by other methods. We measure the worst-case user throughput and all users' throughput in $1000$ random realizations of the network and plot the cumulative distribution functions (CDFs) of them in Fig. \ref{fig:plot_min_rate_eval_cdf}.
Fig. \ref{subfig:plot_min_rate_eval_cdf_group_1} shows the CDFs of user throughput achieved by the trained actor in our methods (with legend ``Proposed'') and achieved by heuristic schemes, RAND and UNIF, and the MC-based scheme. 
The results show that the RAND heuristic scheme performs the worst since it randomly allocates users in RAW slots and can possibly allocate highly contended or interfered users into the same slot. Also, the UNIF and MC-based schemes have close performance because both methods have a strategy that balances the number of transmissions in each slot. Further, the proposed scheme performs the best, e.g., $65\%\sim100\%$ higher worst-case user throughput on average. This is due to the fine exploitation of the contention and interference information in the optimized edge weights. 

Also, we compare the CDF of user throughput achieved by our scheme and achieved by graph-max-cut-based methods in Fig.~\ref{subfig:plot_min_rate_eval_cdf_group_2}. The results show that the MCON and MHID schemes achieve higher worst-case user throughput than the other graph-max-cut-based method, MINT. 
This is because MCON and MHID aim to cut more edges between either contending or hidden users (the two main causes of throughput starvation), while these features are not considered in MINT. 
Also, it is shown that our proposed methods achieve better worst-case user throughput than the MCON and MHID schemes (around $30\%$ on average). This is because our scheme considers the path losses from users to APs as the input when deciding the edge weights, whose values thus have more information on the contention and interference in the network, e.g., the interference power/duration.
Furthermore, the CDFs of all users' throughput in Fig. \ref{subfig:plot_min_rate_eval_cdf_group_1} and Fig. \ref{subfig:plot_min_rate_eval_cdf_group_2} show that the improvement in the worst-case user throughput has a minor reduction in users' throughput above the median. 
We average the total user throughput in the above simulations in Fig. \ref{fig:plot_min_rate_eval_cdf}, and the results indicate that our scheme has $3\%\sim4\%$ less total throughput than other methods, which is negligible compared to the improvement of the worst-case user throughput (e.g., over $65\%$ and over $30\%$ improvement in Fig.~\ref{subfig:plot_min_rate_eval_cdf_group_1} and Fig.~\ref{subfig:plot_min_rate_eval_cdf_group_2}, respectively).

\begin{figure}[t]
\begin{subfigure}[b]{1\columnwidth}
\centering
\includegraphics[scale=0.8]{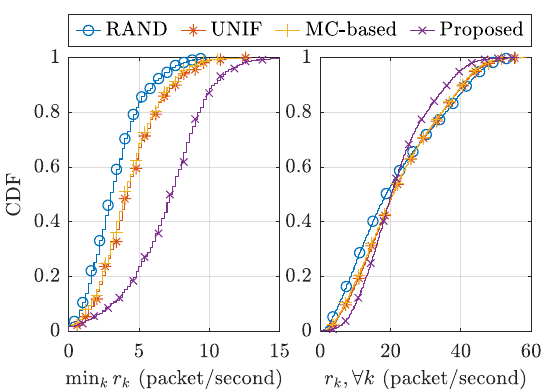}
\caption{Comparison with heuristic and Markov-model-based methods}
\label{subfig:plot_min_rate_eval_cdf_group_1}
\end{subfigure}\\
\begin{subfigure}[b]{1\columnwidth}
\centering
\includegraphics[scale=0.8]{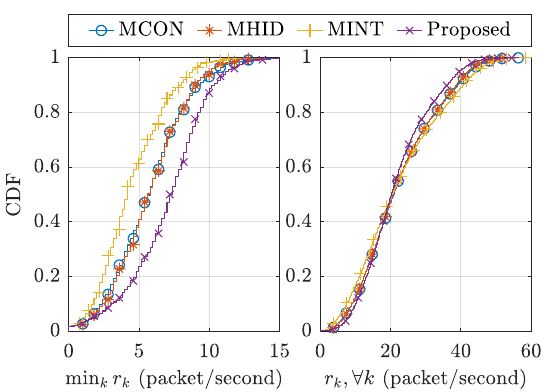}
\caption{Comparison with graph-max-cut-based methods}
\label{subfig:plot_min_rate_eval_cdf_group_2}
\end{subfigure}
\caption{Comparison of the proposed and other methods.}
\label{fig:plot_min_rate_eval_cdf}
\vspace{-0.2cm}
\end{figure}

Fig. \ref{fig:plot_min_rate_loc} shows the locations of users with low throughput with respect to the location of APs in four schemes, namely RAND, MC-based, MHID, and our scheme. 
Specifically, we show the locations of users whose throughput is in the range of $[0,2.5)$ and $[2.5,5)$ in the $1000$ random realizations of the network. The results show that our methods significantly reduce the number of users with low throughput and improve the users' throughput far from the APs, e.g., on the boundary of the simulated area.
\begin{figure}[!t]
\centering
\includegraphics[scale=0.8]{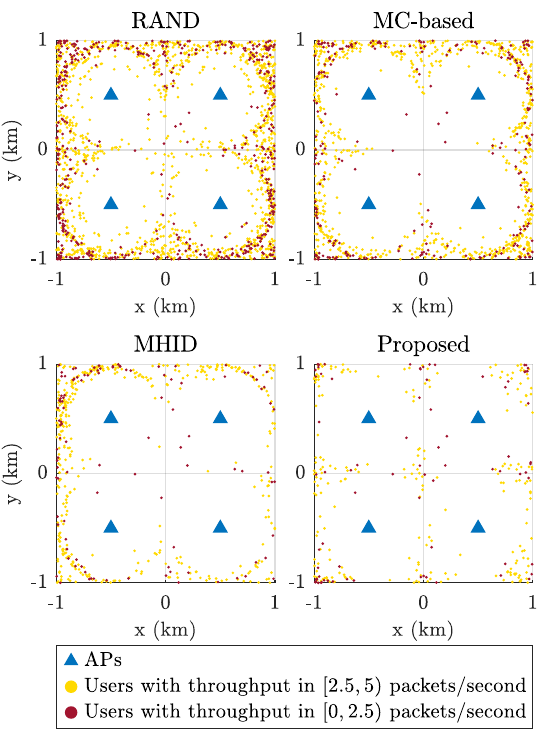}
\caption{Low-throughput users' locations in the simulated area.}
\label{fig:plot_min_rate_loc}
\vspace{-0.2cm}
\end{figure}

\subsubsection{Impact of the Number of Users and Groups}
\begin{figure}[!ht]
\centering
\includegraphics[scale=0.8]{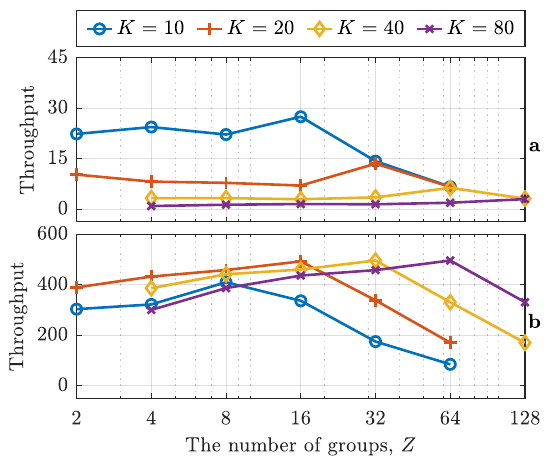}
\caption{The performance of the trained actor with different numbers of users and groups, $K$ and $Z$, (a) the worst-case user throughput and (b) the total throughput in packets per second.}
\label{fig:plot_min_tot_for_k_z}
\vspace{-0.2cm}
\end{figure}

We also study the impact of different numbers of users and groups on the performance of our scheme's grouping decisions. 
Fig.~\ref{fig:plot_min_tot_for_k_z} shows the user throughput when we vary the number of groups, $Z$, from $2$ to $128$ for specific numbers of users, $K$. We repeat the simulation for different $K$, e.g., $10$, $20$, $40$ and $80$. For each case of $K$ and $Z$, we average the measured throughput in $1000$ random realizations of the network

Fig. \ref{fig:plot_min_tot_for_k_z}a shows three different regions in worst-case user throughput measurements when varying $Z$ for each specific $K$. 
First, when $Z$ is smaller than $K$, our scheme achieves approximately the same worst-case user throughput in different $Z$ for given $K$. Note that the increasing number of groups leads to a higher interval between each user's scheduled RAW slots, resulting in fewer scheduled slots (or less transmission time) on average. As the worst-case user throughput is approximately the same when increasing $Z$ (even though there is less transmission time), it implies that our method can exploit additional groups to reduce contention and interference. 
Second, when $Z$ increases to a larger amount than given $K$, the worst-case user throughput is larger than the one in the first region. This is because there are a sufficient number of groups so that almost every user can have an individual slot, where contention or interference barely happens. 
Third, the further increment on $Z$ for given $K$ reduces the worst-case user throughput.
This is because each user transmits at its maximum capacity in each of its scheduled slots without contention and interference, while the high interval between the user's scheduled slots results in low throughput on average.
Furthermore, for given $Z$, increasing $K$ reduces the worst-case user throughput. 
This is because wireless resources for each user's transmissions reduce as the total number of users increases.

The results in Fig. \ref{fig:plot_min_tot_for_k_z}b show two regions in the total user throughput when varying $Z$ in each case of $K$. Specifically, when $Z$ is smaller than $K$, increasing $Z$ for given $K$ improves the total user throughput. This is because more user groups help manage the contention and interference across the network, reducing the backoff time and packet losses. However, when $Z$ is larger than given $K$, increased groups result in a reduced total user throughput. This is because each slot has either few or no users, and the slot's time is not well-utilized for transmissions.
Additionally, we note that the maximum total throughput is close for different numbers of users. This is because a given Wi-Fi network has limited wireless resources, and thus, it has a finite total user throughput.

By comparing Fig. \ref{fig:plot_min_tot_for_k_z}a and \ref{fig:plot_min_tot_for_k_z}b, we observe that the optimal worst-case and total user throughput is achieved at different values of $Z$ in each case of $K$. This indicates that we can establish a trade-off between the worst-case user throughput and the total user throughput by controlling the number of groups, e.g., by setting $Z$ to the value that maximizes either worst-case or total user throughput according to $K$. 
Also, note that increasing the number of groups causes a larger interval between two scheduled slots of a user, leading to higher delays for the packets arriving between scheduled slots. This issue must be considered for latency-sensitive applications. 
Thus, how to optimize the number of groups requires further study.

\subsection{Performance of Proposed Online Architecture}
We then evaluate the performance of the proposed online architecture that fine-tunes the edge weights.
We measure the user throughput and update the weights as Section \ref{sec:online_arch} every $2$ seconds (i.e., every $200$ RAW slots or $\Tilde{T}=200$) for $200$ seconds. 
We fine-tune the edge weights generated by the trained NNs in networks with $10$, $20$ and $40$ users.
We compute the ratio of the worst-case user throughput achieved by fine-tuned edge weights to the one achieved by fixed edge weights as its initial values, as shown in Fig. \ref{fig:plot_online_mob_0_and_mob_2}.
We consider two scenarios on user mobility: 1) static users without movement, as shown in Fig. \ref{fig:plot_online_mob_0_and_mob_2}a, 2) mobile users moving in a random direction at $2$ meters per second, e.g., at approximately human or robot walking speed, (if a user reaches the boundary of the simulated area, it will choose another random direction towards the inside of the area), as shown in Fig. \ref{fig:plot_online_mob_0_and_mob_2}b.

When users' locations are static, the results show that the proposed architecture can further improve $5\sim10\%$ of the worst-case user throughput compared to its initial value when $20$ users are in the network. This is because the online architecture can fine-tune the edge weight according to the feedback on the worst-case user for every $\Tilde{T}$ slots.
Also, the results show that when the number of users is small (e.g., $10$ users), there is no significant improvement, while a large improvement, $5\sim30\%$, of the worst-case user throughput is achieved when more users are in the network  (e.g., $40$ users). This is because when there are fewer users, users are sparsely located in the area, causing less contention and interference with each other. Thus, the margin of system performance is small.
Meanwhile, when more users exist in the network, they are more densely located and make heavy contention and interference in the network, where fine-tuning the edge weights can significantly improve the network performance.

When users are mobile, the improvement achieved by fine-tuning edge weights is significant (around $200\%$) compared to the fixed edge weights when $K$ is $20$ and $40$. This is because the fixed edge weights computed based on initial network states (user-to-AP path losses) are no longer validated when users are moving over time, while the proposed architecture keeps updating the edge weights according to the latest states for every $\Tilde{T}$ slots.
Also, when $K$ is as small as $10$, no significant improvement can be observed in the simulated time. This is because the margin of system performance is small due to little contention and interference in the network, as mentioned before.

\begin{figure}[!t]
\centering
\includegraphics[scale=0.8]{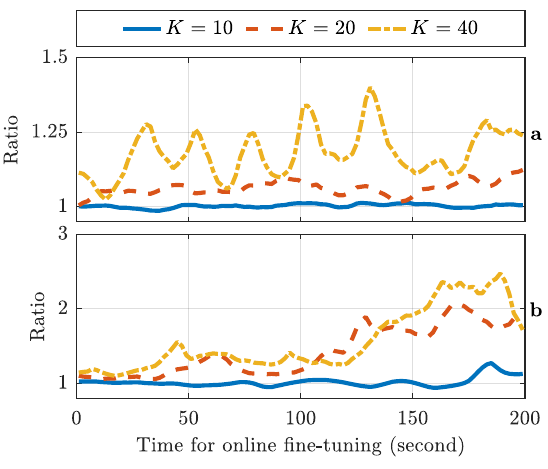}
\caption{Ratios of the worst-case user throughput when edge weights are fine-tuned over time to the one when edge weights are fixed as their initial values for different numbers of users $K$ in the network (a) with static users and (b) with users moving at $2$ meters per second.}
\label{fig:plot_online_mob_0_and_mob_2}
\vspace{-0.2cm}
\end{figure}

\section{Discussion on Limitations}
We discuss the limitations of our methods in this section.
\subsection{Limitation in Generalization}
Although we have shown how to apply the NNs trained in one scenario to other scenarios with different numbers of users/groups or user mobility, the generalization capability of our methods can still be improved. Specifically, we have assumed the number of APs is fixed in this work, which results in a fixed-dimensional network state for each user.
Under this assumption, we design the inference NN $\omega$ as a FNN with a fixed input dimension.
However, when the number of APs changes in the network, e.g., due to power outages or network expansion, $\omega$ cannot adapt to such a change. To address this issue, we can construct $\omega$ using NNs with flexible structures (e.g., GNNs that can take each user's path losses as features on a graph flexible to the number of APs.) and train it in networks with varying numbers of APs.
Additionally, the users and APs are assumed to be homogeneous in this work, while their configurations in practice can be heterogeneous, e.g., in packet rates, transmission power, receiver sensitivity and bandwidth. These features can be included as the network states and the actor and critic can be redesigned to process these states. 
Nevertheless, the above extensions do not alter the fundamental concept of the proposed AC-GRL.

\subsection{Limitation in Complexity}
Our methods have a polynomial complexity in terms of the number of users, as discussed in Section \ref{subsec:grl_algorithm_flow}. It can consume significant computation resources or computing time when the network has a large amount of users. The high complexity is because the graph that we construct is fully connected, where operations are needed for all edges.
Research can be done to reduce the complexity of our algorithm. Note that the complexity of our methods mainly consists of two parts, 1) NNs' forward propagation and backpropagation on each edge of the graph and 2) the max-cut SDP on the graph. This graph can be sparsified \cite{satuluri2011local} by removing the edges between any two users that are not either contending or interfering with each other directly, e.g., they are geographically separated at a large distance.
Consequently, we can reduce the times of NNs' forward propagation and backpropagation \cite{li2021adaptivegcn} as well as the complexity of the max-cut SDP \cite{arora2007combinatorial}.

\section{Conclusion}
In this paper, we studied how to use the RAW mechanism in Wi-Fi HaLow to improve the worst-case user throughput.
We proposed the framework that formulates the user grouping problem in the RAW slot assignment as the graph construction problem. Here, the graph's edge weights are adjusted to represent the contention and interference in each user pair, and the graph's max cut can obtain the user grouping decisions. We developed the AC-GRL algorithm to train NNs that generate the optimal edge weights based on users' path losses measured at AP. Further, we designed the architecture to fine-tune the edge weights generated by trained NNs according to online feedbacks. Simulation results show that our approach achieves much better worst-case user throughput than the existing approaches. 
Also, our online architecture can further improve worst-case user throughput by fine-tuning the NN-generated edge weights and can keep updating the edge weights according to varying network states. The limitations of this work and future research directions are also discussed.

\section*{Appendix: The Proof of Lemma \ref{lemma:equivalence_of_user_grouping_problem}}
\begin{proof}
We prove the statement by contradiction. Suppose the statement is false, which means there exist grouping decisions, $\mathbf{z}'$, maximizing the objective in \eqref{eq:prob:user_grouping} other than $\mathbf{z}^*$. Then, we construct edge weights as $W'_{i,j}$, $\forall i\neq j$, where $W'_{i,j} = \mathbf{1}_{\{z'_i\neq z'_j\}}$. Then, $\mathbf{z}'$ is the optimal solution that maximizes the LLP of \eqref{eq:prob:user_grouping:graph_cut:adaptive_edge_weighting} for the given edge weights as $W'_{i,j}$, $\forall i\neq j$, which means that $\mathbf{z}'$ is the grouping decisions if the edge weights are $W'_{i,j}$, $\forall i\neq j$.
Note that the network objectives in \eqref{eq:prob:user_grouping} and \eqref{eq:prob:user_grouping:graph_cut:adaptive_edge_weighting} are the same and only depend on the grouping decisions. Thus, the values of the objective achieved by $\mathbf{z}'$ in \eqref{eq:prob:user_grouping} and $W'_{i,j}$, $\forall i\neq j$, in \eqref{eq:prob:user_grouping:graph_cut:adaptive_edge_weighting} are the same because their grouping decisions are the same. It also implies that the value of the objective achieved by $\mathbf{z}^*$ in \eqref{eq:prob:user_grouping} is equal to the one achieved by $W^*_{i,j}$, $\forall i\neq j$, in \eqref{eq:prob:user_grouping:graph_cut:adaptive_edge_weighting} since $\mathbf{z}^*$ is the grouping decisions by cutting $W^*_{i,j}$, $\forall i\neq j$. We note that $W^*_{i,j}$, $\forall i\neq j$, maximizes \eqref{eq:prob:user_grouping:graph_cut:adaptive_edge_weighting}, which means the objective in \eqref{eq:prob:user_grouping:graph_cut:adaptive_edge_weighting} achieved by $W^*_{i,j}$, $\forall i\neq j$, is greater than or equal to the one in \eqref{eq:prob:user_grouping:graph_cut:adaptive_edge_weighting} achieved by $W'_{i,j}$, $\forall i\neq j$. This implies that the objective achieved by $\mathbf{z}^*$ in \eqref{eq:prob:user_grouping} is greater than or equal to the one achieved by $\mathbf{z}'$ in \eqref{eq:prob:user_grouping}, which is contradictory to the assumption on the optimality of $\mathbf{z}'$ and the non-optimality of $\mathbf{z}^*$ in \eqref{eq:prob:user_grouping} at the start of the proof. Therefore, the statement in Lemma \ref{lemma:equivalence_of_user_grouping_problem} is true.
\end{proof}

\section*{Appendix: Complexity Analysis of AC-GRL}
We first analyze the complexity of FNNs and GCNs that construct the actor and the critic, as well as the complexity of the max-cut SDP in \eqref{eq:prob:sub_cut_sdp}. 
The FNN's complexity in either forward propagation or backpropagation is quadratic to its input and output layer dimensions (assuming hidden layer dimensions are multiple of either input or output dimensions, e.g., the NN configurations in this work) \cite{goodfellow2016deep}. In our case, FNNs have constant dimensions regardless of $A$ and $K$, including $\dot{\mu}$, $\mathrm{EFE}$, $\mathrm{NFE}$, $\mathrm{HNFE}^l,\forall l$ and $\mathrm{ROF}$, while exceptionally the inference NN $\omega$ has dimensions linear to $A$. For FNNs with constant dimensions, we write their complexity as $\mathcal{O}(1)$ for each forward propagation or backpropagation. Meanwhile, we write the complexity of $\omega$ as $\mathcal{O}(A^2)$. The complexity of the GCN in \eqref{eq:gcn} is $\mathcal{O}(K\cdot M^2 + K^2\cdot M)\approx\mathcal{O}(K^2)$ in either forward propagation or backpropagation \cite{kipf2017semi,liu2020efficient}. The complexity of the max-cut SDP in \eqref{eq:prob:sub_cut_sdp} depends on the convex optimizer's implementation, and the typical complexity is $\mathcal{O}(|\mathcal{V}^{\beta,c}|^{3.5})$ \cite{boyd2004convex,bubeck2015convex,o2016conic,toh1999sdpt3} for a given (sub)set $\mathcal{V}^{\beta,c}$. Because each (sub)set has at most $K$ users, the SDP in \eqref{eq:prob:sub_cut_sdp} has a complexity upper-bounded in $\mathcal{O}(K^{3.5})$.

Based on the above analysis, we then provide the complexity of the AC-GRL algorithm. Each pre-training step of the inference NN $\omega$ in lines \ref{alg:line:training_inference_nn_start}-\ref{alg:line:training_ac_nn_end} costs $K^2$ times of $\omega$'s forward propagations and backpropagation, which has a complexity $K^2\cdot\mathcal{O}(A^2) \approx \mathcal{O}(K^2 A^2)$.
In each step of the main process in lines \ref{alg:line:training_ac_nn_start}-\ref{alg:line:training_ac_nn_end}, the actor generates the edge weights for each user pair, which involves $K^2$ repetitions of the actor's forward propagation using $\dot{\mu}$ and $\omega$, which has a complexity $K^2\cdot(\mathcal{O}(1) + \mathcal{O}(A^2)) \approx \mathcal{O}(K^2 A^2)$. Next, the graph cut procedure divides users into $Z$ groups using \eqref{eq:z_by_DoGraphCut}, which requires solving $Z-1$ (or $\sum_{\beta=0}^{\log_2(Z)-1}2^\beta$) times of the SDP in \eqref{eq:prob:sub_cut_sdp}. Thus, the graph cut procedure (where the SDP's complexity dominates) has a complexity $\mathcal{O}( (Z-1) K^{3.5})$. The update of the critic's parameters first requires one forward propagation and backpropagation of the critic. The critic uses $K^2$, $K$, $K$ and $K$ times of $\mathrm{EFE}$, $\mathrm{NFE}$, $\mathrm{HNFE}^l, \forall l$ and $\mathrm{ROF}$, respectively. Additionally, the critic uses $\zeta E$ GCNs. Thus, the complexity of the update of the critic is $\mathcal{O}(K^2+K+K\zeta+K+K^2\zeta E)\approx \mathcal{O}(K^2)$. The update of the actor requires one forward propagation and backpropagation of the critic as well as $K^2$ backpropagation of $\dot{\mu}$, which has a complexity $\mathcal{O}(K^2 + K^2)\approx \mathcal{O}(K^2)$ in total. Note that $\omega$ is not counted in the complexity of the actor and critic's update because we can use its outputs computed in the actor's weight generation, and we do not update $\omega$'s parameters here.
\bibliography{ref}
\bibliographystyle{IEEEtran}

\begin{IEEEbiography}
[{\includegraphics[width=1in,height=1.25in,clip,keepaspectratio]{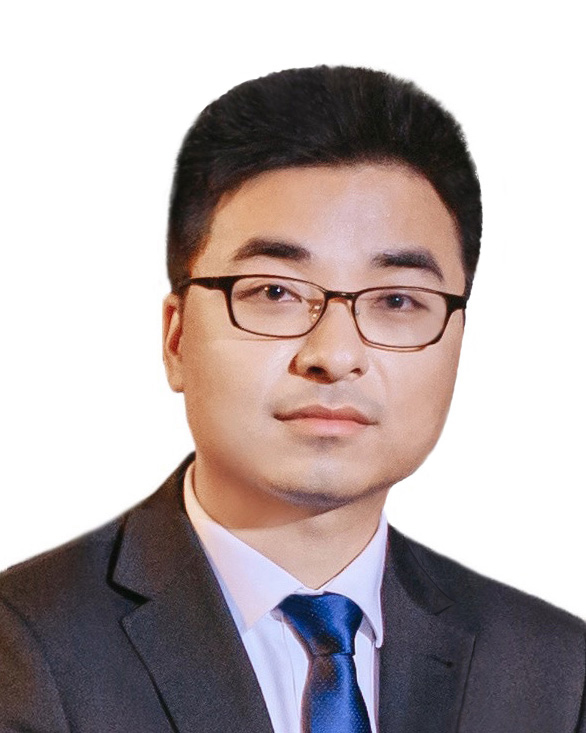}}]{Zhouyou Gu} received the B.E. (Hons.) and M.Phil.
degrees from The University of Sydney (USYD), Australia,
in 2016 and 2019, respectively, where he completed his Ph.D. degree with the School of Electrical and Information Engineering in 2023. 
He was a research assistant at the Centre for IoT and Telecommunications at USYD. His research interests include
designs of real-time schedulers, programmability, and graph and machine learning methods in wireless networks.
\end{IEEEbiography}

\begin{IEEEbiography}
[{\includegraphics[width=1in,height=1.25in,clip,keepaspectratio]{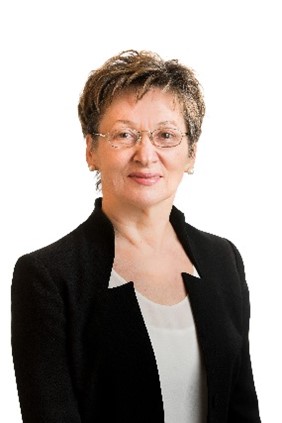}}]{Branka Vucetic} (Life Fellow, IEEE) received the B.S., M.S., and Ph.D. degrees in electrical engineering from the University of Belgrade, Belgrade, Serbia. She is an Australian Laureate Fellow, a Professor of Telecommunications, and the Director of the Centre for IoT and Telecommunications, the University of Sydney, Camperdown, NSW, Australia. Her current research work is in wireless networks and Industry 5.0. In the area of wireless networks, she works on communication system design for 6G and wireless AI. In the area of Industry 5.0, her research is focused on the design of cyber–physical human systems and wireless networks for applications in healthcare, energy grids, and advanced manufacturing. She is a Fellow of the Australian Academy of Technological Sciences and Engineering and the Australian Academy of Science.
\end{IEEEbiography}

\begin{IEEEbiography}
[{\includegraphics[width=1in,height=1.25in,clip,keepaspectratio]{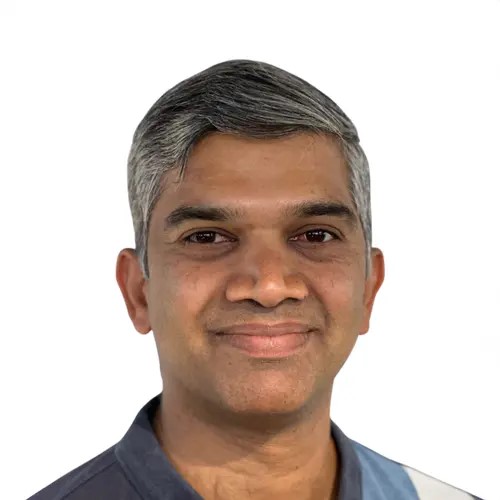}}]{Kishore Chikkam} is the VP Systems \& Software at Morse Micro. Kishore has two decades of experience in systems design for wireless local area network chips. Before joining Morse Micro, Kishore worked for large semiconductor companies, including Athena Semiconductors, LifeSignals, and Broadcom. Kishore leads the systems and software teams from design to deployment in the field to ensure the world-leading performance of Morse Micro chips.
\end{IEEEbiography}

\begin{IEEEbiography}
[{\includegraphics[width=1in,height=1.25in,clip,keepaspectratio]{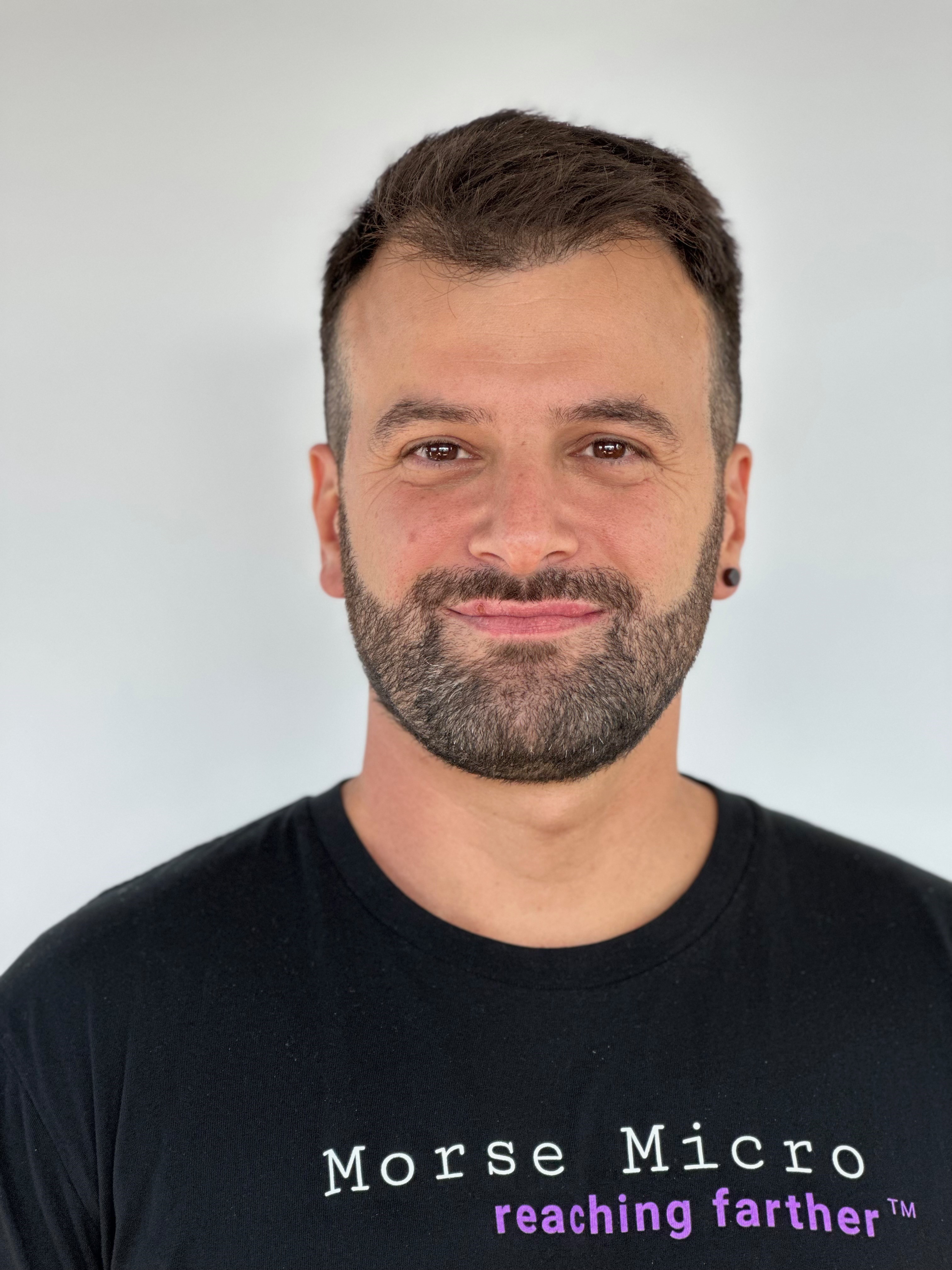}}]{Pasquale Aliberti} is the Chief of Staff at Morse Micro. Pasquale manages collaborations and intellectual properties at Morse Micro. Before joining Morse Micro, Pasquale worked at Silicon Quantum Computing and at the Photovoltaics Centre of Excellence for Photovoltaics at the University of New South Wales.
\end{IEEEbiography}

\begin{IEEEbiography}
[{\includegraphics[width=1in,height=1.25in,clip,keepaspectratio]{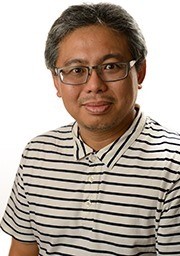}}]{Wibowo Hardjawana} (Senior Member, IEEE) received a PhD in electrical engineering from The University of Sydney, Australia. He is currently a Senior Lecturer in Telecommunications with the School of Electrical and Information Engineering at the University of Sydney. Before that, he was with Singapore Telecom Ltd., managing core and radio access networks. His current fundamental and applied research interests are in AI applications for 5/6G cellular radio access and WiFi networks. He focuses on system architectures, resource scheduling, interference, signal processing, and the development of wireless standard-compliant prototypes. He has also worked with several industries in the area of 5G and long-range WiFi. He was an Australian Research Council Discovery Early Career Research Award Fellow.
\end{IEEEbiography}

\end{document}